\documentclass[onecolumn,draftclsnofoot,11pt]{IEEEtran}%[onecolumn,draft]{IEEEtran}%[journal]{IEEEtran}
\usepackage{dsfont,graphicx,color,epsf,epsfig,verbatim,amssymb,amsmath,amsthm}
\usepackage{latexsym,amsfonts,caption,array,cite,subfigure,multicol,multirow}

\setkeys{Gin}{draft=false}

\newtheorem{lemma}{Lemma}[section]
\newtheorem{remark}{Remark}[section]

\newcommand{\G}{\mathbf{G}}
\newcommand{\R}{\mathbf{R}}

\DeclareMathOperator*{\argmax}{arg\,max}

\begin{document}
\title{Multilevel Polarization of Polar Codes Over Arbitrary Discrete Memoryless Channels}% can use linebreaks \\
\author{\IEEEauthorblockN{Aria G. Sahebi and S. Sandeep Pradhan\\ \thanks{This work was presented in part in the 49th Annual Allerton Conference, Allerton, IL, USA, September 28 - 30, 2011.} \thanks{This work was supported by NSF grants CCF-0915619 and CCF-1116021.}}
\IEEEauthorblockA{Department of Electrical Engineering and Computer Science,\\
University of Michigan, Ann Arbor, MI 48109, USA.\\
Email: \tt\small ariaghs@umich.edu, pradhanv@umich.edu}}

%\markboth{Draft}
%{Sahebi \MakeLowercase{\textit{et al.}}: Multilevel Polarization of Polar Codes Over Arbitrary Discrete Memoryless Channels}

\maketitle

\begin{abstract}
It is shown that polar codes achieve the symmetric capacity of discrete memoryless channels with arbitrary input alphabet sizes. It is shown that in general, channel polarization happens in several, rather than only two levels so that the synthesized channels are either useless, perfect or ``partially perfect''. Any subset of the channel input alphabet which is closed under addition, induces a coset partition of the alphabet through its shifts. For any such partition of the input alphabet, there exists a corresponding partially perfect channel whose outputs uniquely determine the coset to which the channel input belongs. By a slight modification of the encoding and decoding rules, it is shown that perfect transmission of certain information symbols over partially perfect channels is possible. Our result is general regarding both the cardinality and the algebraic structure of the channel input alphabet; i.e we show that for any channel input alphabet size and any Abelian group structure on the alphabet, polar codes are optimal. It is also shown through an example that polar codes when considered as group/coset codes, do not achieve the capacity achievable using coset codes over arbitrary channels.%As an example, for polar codes over $\mathds{Z}_6$, the asymptotic synthesized channels can exist in four forms: 1) can truly determine which one of the cosets $\{0\}$, $\{1\}$, $\{2\}$, $\{3\}$, $\{4\}$ or $\{5\}$ contains the input symbol, (perfect channels with capacity $\log_2 6$ bits per channel use), 2) can truly determine which one of the cosets $\{0,3\}$, $\{1,4\}$ or $\{2,5\}$ contains the input symbol (partially perfect channels with capacity $\log_2 3$ bits per channel use), 3) can truly determine which one of the cosets $\{0,2,4\}$ or $\{1,3,5\}$ contains the input symbol (partially perfect channels with capacity $1$ bit per channel use), 4) can only determine the input belongs to $\{0,1,2,3,4,5\}$ (useless channel). Cases 1,2,3 and 4 correspond to the coset decomposition of $\mathds{Z}_6$ based on subgroups $\{0\}$, $\{0,3\}$, $\{0,2,4\}$ and $\{0,1,2,3,4,5\}$ respectively.
\end{abstract}


\IEEEpeerreviewmaketitle

\section{Introduction}
\IEEEPARstart{P}olar codes were originally proposed by Arikan in \cite{arikan_polar} for discrete memoryless channels with a binary input alphabet. Polar codes over binary input channels are shifted linear (coset) codes capable of achieving the symmetric capacity of channels. These codes are constructed based on the Kronecker power of the $2\times 2$ matrix $\left[\begin{array}{cc}1&0\\1&1\\\end{array}\right]$ and are the first known class of capacity achieving codes with an explicit construction.\\

It is known that non-binary codes outperform binary codes in certain communication settings. Therefore, constructing capacity achieving codes for channels of arbitrary input alphabet sizes is of great interest. In order to construct capacity achieving codes over non-binary channels, there have been attempts to extend polar coding techniques to channels of arbitrary input alphabet sizes. It is shown in \cite{sasoglu_polar_q} that polar codes achieve the symmetric capacity of channels when the size of the input alphabet is a prime. For channels of arbitrary input alphabet sizes, it is shown in \cite{sasoglu_polar_q} that the original construction of polar codes does not necessarily achieve the symmetric capacity of the channel due to the fact that polarization (into two levels) may not occur for arbitrary channels. In the same paper, a randomized construction of polar codes based on permutations is proposed. In this approach, the existence of a polarizing transformation is shown by a (small) random coding argument over the ensemble of permutations of the input alphabet. In another approach in \cite{sasoglu_polar_q}, a code construction method is proposed which is based on the decomposition of the composite input channel into sub-channels of prime input alphabet sizes. In this multilevel code construction method, a separate polar code is designed for each sub-channel of prime input alphabet size. It is shown in \cite{mori_polar} that for channels for which the input alphabet size is a prime power, polar codes defined on the input alphabet can achieve the symmetric capacity without the need to use multilevel code construction methods.\\

Another related work is \cite{Abbe_Polar_Mac}, in which the authors have shown that polar codes are sufficient to achieve the uniform sum rate on any binary input MAC and it is stated that the same technique can be used for the point-to-point problem to achieve the symmetric capacity of the channel when the size of the alphabet is a power of $2$. In a recent work, it has been shown in \cite{Park_Barg_Polar} that polar codes achieve the capacity of channels with input alphabet size a power of $2$.\\

In this paper, we show that with a slight modification of the encoding and decoding rules, standard polar codes are sufficient to achieve the symmetric capacity of all discrete memoryless channels. Our result is general regarding both the cardinality and the algebraic structure of the channel input alphabet; i.e we show that for any channel input alphabet size and any Abelian group structure on the alphabet, polar codes are optimal. This result was first reported in \cite{Sahebi_polar_allerton2011}. We use a combination of algebraic and coding techniques and show that in general, channel polarization occurs in several levels rather than only two: Suppose the channel input alphabet is $\G$ and is endowed with an Abelian group structure. Then for any subset $H$ of the channel input alphabet $\G$ which is closed under addition (i.e any subgroup of $\G$), there may exist a corresponding polarized channel which can perfectly transmit the index of the shift (coset) of $H$ in $\G$ which contains the input. As an example, for a channel of input $\mathds{Z}_6$, there are four subgroups of the input alphabet: \textbf{i)} $\{0\}$ with cosets $\{0\}$, $\{1\}$, $\{2\}$, $\{3\}$, $\{4\}$ and $\{5\}$, \textbf{ii)} $\{0,3\}$ with cosets $\{0,3\}$, $\{1,4\}$ and $\{2,5\}$, \textbf{iii)} $\{0,2,4\}$ with cosets $\{0,2,4\}$ and $\{1,3,5\}$ and \textbf{iv)} $\mathds{Z}_6$. For polar codes over $\mathds{Z}_6$, the asymptotic synthesized channels can exist in four forms: \textbf{i)} can determine which one of the cosets $\{0\}$, $\{1\}$, $\{2\}$, $\{3\}$, $\{4\}$ or $\{5\}$ contains the input symbol, (perfect channels with capacity $\log_2 6$ bits per channel use), \textbf{ii)} can determine which one of the cosets $\{0,3\}$, $\{1,4\}$ or $\{2,5\}$ contains the input symbol (partially perfect channels with capacity $\log_2 3$ bits per channel use), \textbf{iii)} can determine which one of the cosets $\{0,2,4\}$ or $\{1,3,5\}$ contains the input symbol (partially perfect channels with capacity $1$ bit per channel use), \textbf{iv)} can only determine the input belongs to $\{0,1,2,3,4,5\}$ (useless channel). Cases \textbf{i},\textbf{ii},\textbf{iii} and \textbf{iv} correspond to coset decompositions of $\mathds{Z}_6$ based on subgroups $\{0\}$, $\{0,3\}$, $\{0,2,4\}$ and $\{0,1,2,3,4,5\}$ respectively.\\% The fact that subgroups of the channel input alphabet ar\\% The algebraic nature of polar codes is in agreement with the fact that subgroups of the underlying group can be present in asymptotic channels.\\

Although standard binary polar codes are group (linear) codes, the class of capacity achieving codes constructed and analyzed in this paper are not group codes. It is known that group codes do not generally achieve the symmetric capacity of discrete memoryless channels \cite{ahlswede_group}. Hence, one could have predicted that standard polar codes cannot achieve the symmetric capacity of arbitrary channels and a modification of the encoding rule is indeed necessary to achieve that goal. Due to the modifications we make to the encoding rule of polar codes, the constructed codes fall into a larger class of structured codes called nested group codes.\\

The paper is organized as follows: In Section \ref{prel}, some definitions and basic facts are stated which are used in the paper. In Section \ref{section:example}, we present two motivating examples of $4$-ary and $6$-ary channels and observe the polarization effect on these channels. In Section \ref{rings}, we show that polar codes achieve the symmetric capacity of channels with input alphabet size $q=p^r$ where $p$ is a prime and $r$ is an integer. This result is generalized to arbitrary channels in Section \ref{section:abelian}. In Section \ref{section:examples}, the relation of polar codes to group codes is discussed and two examples of channels over $\mathds{Z}_4$ are provided. In the first example, we show that polar codes approach the capacity of channels achievable using group codes. The intent of the second example is to show that this is not generally the case; i.e. polar codes do not generally approach the capacity of channels achievable using group/coset codes.\\

\section{Preliminaries} \label{prel}
\subsubsection{Source and Channel Models}
We consider discrete memoryless and stationary channels used without feedback. We associate two finite sets $\mathcal{X}$ and $\mathcal{Y}$ with the channel as the channel input and output alphabets. These channels can be characterized by a conditional probability law $W(y|x)$ for $x\in \mathcal{X}$ and $y\in \mathcal{Y}$. The channel is specified by $(\mathcal{X},\mathcal{Y},W)$. The source of information generates messages over the set $\{1,2,\ldots,M\}$ uniformly for some positive integer $M$.\\

\subsubsection{Achievability and Capacity}
A transmission system with parameters $(n,M,\tau)$ for reliable communication over a given channel $(\mathcal{X},\mathcal{Y},W)$ consists of an encoding mapping $e:\{1,2,\ldots,M\}\rightarrow \mathcal{X}^n$ and a decoding mapping $d:\mathcal{Y}^n\rightarrow\{1,2,\ldots,M\}$ such that
\begin{align*}
\frac{1}{M}\sum_{m=1}^{M}W^n\left(d(Y^n)\ne
m|X^n=e(m)\right)\le \tau
\end{align*}
Given a channel $(\mathcal{X},\mathcal{Y},W)$, the rate $R$ is said to be achievable if for all $\epsilon>0$ and for all sufficiently large $n$, there exists a transmission system for reliable communication with parameters $(n,M,\tau)$ such that
\begin{align*}
\frac{1}{n}\log M \ge R-\epsilon,\qquad\qquad \tau\le \epsilon
\end{align*}
\\
\subsubsection{Symmetric Capacity and the Bhattacharyya Parameter}
For a channel $(\mathcal{X},\mathcal{Y},W)$, the symmetric capacity is defined as $I^0(W)=I(X;Y)$ where the channel input $X$ is uniformly distributed over $\mathcal{X}$ and $Y$ is the output of the channel; i.e. for $q=|\mathcal{X}|$,
\begin{align*}
I^0(W)=\sum_{x\in\mathcal{X}}\sum_{y\in\mathcal{Y}}\frac{1}{q}W(y|x)\log\frac{W(y|x)}{\displaystyle\sum_{{\tilde{x} \in\mathcal{X}}}\frac{1}{q}W(y|\tilde{x})}
\end{align*}
The Bhattacharyya distance between two distinct input symbols $x$ and $\tilde{x}$ is defined as
\begin{align*}
Z(W_{\{x,\tilde{x}\}})=\sum_{y\in\mathcal{Y}}\sqrt{W(y|x)W(y|\tilde{x})}
\end{align*}
and the average Bhattacharyya distance is defined as
\begin{align*}
Z(W)=\sum_{\substack{x,\tilde{x}\in \mathcal{X}\\x\ne\tilde{x}}}\frac{1}{q(q-1)}Z(W_{\{x,\tilde{x}\}})
\end{align*}
\\
\subsubsection{Binary Polar Codes}
For any $N=2^n$, a polar code of length $N$ designed for the channel $(\mathds{Z}_2,\mathcal{Y},W)$ is a linear code characterized by a generator matrix $G_N$ and a set of indices $A\subseteq \{1,\cdots,N\}$ of \emph{perfect channels}. The generator matrix for polar codes is defined as $G_N=B_NF^{\otimes n}$ where $B_N$ is a permutation of rows, $F=\left[ \begin{array}{cc}1 & 0\\1 & 1\end{array} \right]$ and $\otimes$ denotes the Kronecker product. The set $A$ is a function of the channel. The decoding algorithm for polar codes is a specific form of successive cancellation \cite{arikan_polar}.\\

\subsubsection{Groups, Rings and Fields}
All groups referred to in this paper are \emph{Abelian groups}. Given a group $(\G,+)$, a subset $H$ of $\G$ is called a \emph{subgroup} of $\G$ if it is closed under the group operation. In this case, $(H,+)$ is a group on its own right. This is denoted by $H\le \G$. A \emph{coset} $C$ of a subgroup $H$ is a shift of $H$ by an arbitrary element $a\in \G$ (i.e. $C=a+H$ for some $a\in\G$). For any subgroup $H$ of $\G$, its cosets partition the group $\G$. A \emph{transversal} $T$ of a subgroup $H$ of $\G$ is a subset of $\G$ containing one and only one element from each coset (shift) of $H$.\\
We give some examples in the following: The simplest non-trivial example of groups is $\mathds{Z}_2$ with addition mod-$2$ which is a \emph{ring} and a \emph{field} with multiplication mod-$2$. The group $\mathds{Z}_2\times \mathds{Z}_2$ is also a ring and a field under component-wise mod-$2$ addition and a carefully defined multiplication. The group $\mathds{Z}_4$ with mod-$4$ addition and multiplication is a ring but not a field since the element $2\in\mathds{Z}_4$ does not have a multiplicative inverse. The subset $\{0,2\}$ is a subgroup of $\mathds{Z}_4$ since it is closed under mod-$4$ addition. $\{0\}$ and $\mathds{Z}_4$ are the two other subgroups of $\mathds{Z}_4$. The group $\mathds{Z}_6$ is neither a field nor a ring. Subgroups of $\mathds{Z}_6$ are: $\{0\}$, $\{0,3\}$, $\{0,2,4\}$ and $\mathds{Z}_6$.\\

\subsubsection{Polar Codes Over Abelian Groups}
For any discrete memoryless channel, there always exists an {Abelian group} of the same size as that of the channel input alphabet. In general, for an Abelian group, there may not exist a multiplication operation. Since polar encoders are characterized by a matrix multiplication, before using these codes for channels of arbitrary input alphabet sizes, a generator matrix for codes over Abelian groups needs to be properly defined. In Appendix \ref{section:polar_abelian}, a convention is introduced to generate codes over groups using $\{0,1\}$-valued generator matrices.\\

\subsubsection{Group Codes}
Let the channel input alphabet $\mathcal{X}$ be equipped with the structure of a finite Abelian group $\G$ of the same size. Then the channel is specified by $(\G,\mathcal{Y},W)$. A group code over $\G$ of length $N$ for this channel is any {subgroup} of $\G^N$. The group capacity of a channel $(\G,\mathcal{Y},W)$ is the maximum achievable rate using group codes over $\G$ for this channel. Group codes generalize the notion of linear codes over {fields} to channels with composite input alphabet sizes. A coset code is a shift of a group code by a constant vector.\\

\subsubsection{Notation}%We say a function $f:\mathds{R}\rightarrow\mathds{R}$ is $O(\epsilon)$ if it's right limit is zero at zero.
We denote by $O(\epsilon)$ any function of $\epsilon$ which is right-continuous around $0$ and that $O(\epsilon)\rightarrow 0$ as $\epsilon\downarrow 0$. We denote by $a\approx_\epsilon b$ to mean $a=b+O(\epsilon)$.\\
For positive integers $N$ and $r$, let $\{A_0,A_1,\cdots,A_r\}$ be a partition of the index set $\{1,2,\cdots,N\}$. Given sets $T_t$ for $t=0,\cdots,r$, the direct sum $\bigoplus_{t=0}^r T_t^{A_t}$ is defined as the set of all tuples $u_1^N=(u_1,\cdots,u_N)$ such that $u_i\in T_t$ whenever $i\in A_t$.\\

\section{Motivating Examples}\label{section:example}
A key property of the basic polarizing transforms used for binary polar codes is that they have perfect and useless channels as their ``fixed points''; in the sense that, if these transforms are applied to a perfect (useless) channel, the resulting channel is also perfect (useless). In the following, we try to demonstrate that for non-binary channels, the basic transforms have fixed points which are neither perfect nor useless. Consider a $4$-ary channel $(\mathds{Z}_4,\mathcal{Y},W)$ and assume the channel is such that $W(y|u)=W(y|u+2)$ for all $y\in \mathcal{Y}$ and all $u\in \mathds{Z}_4$; i.e. the channel cannot distinguish between inputs $u$ and $u+2$. Consider the transformed channels $W^-$ and $W^+$ originally introduced in \cite{arikan_polar} (Refer to Equations \eqref{eqn:channel_transform1} and \eqref{eqn:channel_transform2} of the current paper). It turns out that
\begin{align*}
&W^+(y_1,y_2,u_1|u_2)=W^+(y_1,y_2,u_1|u_2+2)\\
&W^-(y_1,y_2|u_1)=W^-(y_1,y_2|u_1+2)
\end{align*}
for all $y_1,y_2\in \mathcal{Y}$ and all $u_1,u_2\in \mathds{Z}_4$. This observation is closely related to the fact that $\{0,2\}$ is closed under addition mod-$4$; i.e. the fact that $\{0,2\}$ forms a subgroup of $\mathds{Z}_4$. This means that the transformed channels inherit this characteristic feature of the original channel, in the sense that they cannot distinguish between inputs $u_i$ and $u_i+2$ ($i=2$ for $W^+$ and $i=1$ for $W^-$). This suggests that even in the asymptotic regime, the transformed channels can only distinguish between the sets $\{0,2\}$ and $\{1,3\}$, and not within each set. In the following, we give an example for which such cases indeed exist in the asymptotic regime.\\

Consider the channel depicted in Figure \ref{fig:channel}. For this channel, the symmetric capacity is equal to $C=I(X;Y)=2-\epsilon-2\lambda$. Depending on the values of the parameters $\epsilon$ and $\lambda$, this channel can present three extreme cases: 1) If $\lambda =1$, this channel is useless. 2) If $\epsilon=1$, this channel cannot distinguish between inputs $u$ and $u+2$ and has a capacity of $1$ bit per channel use. 3) If $\epsilon=\lambda=0$, this channel is perfect and has a capacity of $2$ bits per channel use.

\begin{figure}[h]
\centering
\includegraphics[scale=.55]{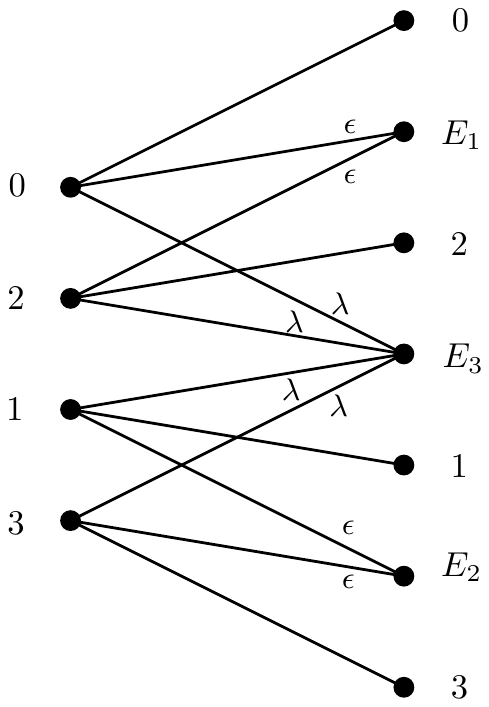}
\caption{\small Channel 1: The input of the channel has the structure of the group $\mathds{Z}_4$. The parameters $\epsilon$ and $\lambda$ take values from $[0,1]$ such that $\epsilon+\lambda\le 1$. $E_1$ and $E_2$ are erasures connected to cosets of the subgroup $\{0,2\}$. The lines connecting the output symbols $0,2,1,3$ to their corresponding inputs, represent a conditional probability of $1-\epsilon-\lambda$. For this channel, the process $I(W^{b_1b_2\cdots b_n})$ can be explicitly found for each $n$ and the multilevel polarization can be observed.}
\label{fig:channel}
\end{figure}
Given a sequence of bits $b_1b_2\cdots b_n$, define $W^{b_1b_2\cdots b_n}$ as in \cite[Section IV]{arikan_polar}, and let $I(W^{b_1b_2\cdots b_n})$ be the mutual information between the input and output of $W^{b_1b_2\cdots b_n}$ when the input is uniformly distributed. We can find $I(W^{b_1b_2\cdots b_n})$ using the following recursion for which the proof can be found in Appendix \ref{appendix:recursion}.\\
Define $\epsilon_0=\epsilon$ and $\lambda_0=\lambda$. For $i=1,\cdots,n$,
\begin{itemize}
\item If $b_i=1$, let
\begin{align}
\label{eqn:recursion1}
\left\{\begin{array}{l}
\epsilon_i=\epsilon_{i-1}^2+2\epsilon_{i-1}\lambda_{i-1}\\
\lambda_i=\lambda_{i-1}^2
\end{array}\right.
\end{align}

\item If $b_i=0$, let
\begin{align}\label{eqn:recursion2}
\left\{\begin{array}{l}
\epsilon_i=2\epsilon_{i-1}-\left(\epsilon_{i-1}^2+2\epsilon_{i-1}\lambda_{i-1}\right)\\
\lambda_i=2\lambda_{i-1}-\lambda_{i-1}^2
\end{array}\right.
\end{align}
\end{itemize}
Then we have $I(W^{b_1b_2\cdots b_n})=2-\epsilon_n-2\lambda_n$.\\
Consider the function $f:[0,1]^2\rightarrow [0,1]^2$, $f(\epsilon,\lambda)=(\epsilon^2+2\epsilon\lambda,\lambda^2)$ corresponding to Equation \eqref{eqn:recursion1}. The fixed points of this function are given by $(0,1)$, $(1,0)$ and $(0,0)$. Similarly, consider the function $g:[0,1]^2\rightarrow [0,1]^2$, $g(\epsilon,\lambda) =(2\epsilon-(\epsilon^2+2\epsilon\lambda),2\lambda-\lambda^2)$ corresponding to Equation \eqref{eqn:recursion2}. It turns out that the fixed points of $g$ are the same as those of $f$. This suggests that in the limit, the transformed channels converge to one of three extreme cases discussed above. Figures \ref{fig:pol2} and \ref{fig:pol} show that it is indeed the case and depicts the three level polarization of the mutual information process $I(W^{b_1b_2\cdots b_n})$ to a discrete random variable $I^{\infty}$ as $n$ grows.

\begin{figure}[h]
\centering
\includegraphics[scale=.30]{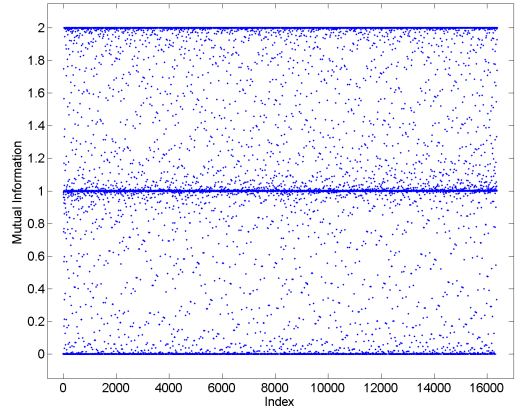}
\caption{\small The behavior of $I(W^{b_1b_2\cdots b_n})$ for $n=14$ for Channel 1 when $\epsilon=0.4$ and $\lambda=0.2$. The three solid lines represent the three discrete values of $I^\infty$ with positive probability.}
\label{fig:pol2}
\end{figure}

\begin{figure}[h]
\centering
\includegraphics[scale=.30]{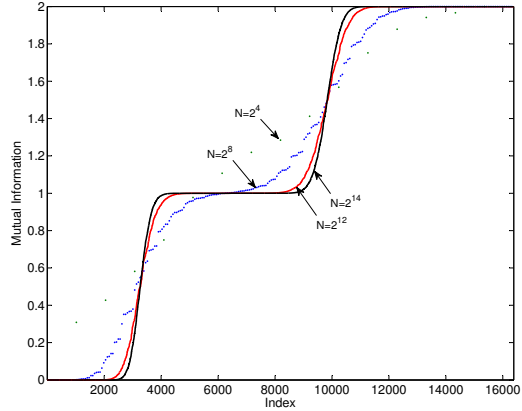}
\caption{\small The asymptotic behavior of $I(W^{b_1b_2\cdots b_n})$, $N=2^n=2^4,2^8,2^{12},2^{14}$ for Channel 1 when the data is sorted. We observe that for this channel, all three extreme cases appear with positive probability. In general, it is possible to have fewer cases in the asymptotic regime.}
\label{fig:pol}
\end{figure}
When $N=2^n$ is large, let $N_0$ be the number of useless channels (corresponding to the width of the first step in Figure \ref{fig:pol}), $N_1$ be the number of partially perfect channels (corresponding to the width of the second step in Figure \ref{fig:pol}) and $N_2$ be the number of perfect channels (corresponding to the width of the third step in Figure \ref{fig:pol}). Since the mutual information process is a martingale, it follows that
\begin{align*}
C=\mathds{E}\{I^{\infty}\}\approx\frac{N_0}{N}\times 0+\frac{N_1}{N}\times 1+\frac{N_2}{N}\times 2
\end{align*}
where $C$ is the symmetric capacity of the channel. Consider the following encoding rule: For indices corresponding to useless channels, let the input symbol take values from $\{0\}$ (from the transversal of the subgroup $\mathds{Z}_4$ of $\mathds{Z}_4$ i.e. fix the input). For indices corresponding to partially perfect channels, let the input symbol take values from $\{0,1\}$ (from the transversal of the subgroup $\{0,2\}$ of $\mathds{Z}_4$). For indices corresponding to perfect channels, let the input symbol take values from $\mathds{Z}_4$ (choose information symbols from the transversal of the subgroup $\{0\}$ of $\mathds{Z}_4$). It turns out that this encoding rule used with an appropriate decoding rule has a vanishingly small probability of error as $N$ becomes large. The rate of this code is equal to
\begin{align*}
R=\frac{1}{N}\left(N_0\log_2 1+N_1\log_2 2+N_2\log_2 4\right)
\end{align*}
This means $R=C$ is achievable using polar codes.\\

Next, we consider a channel with a composite input alphabet size. Consider the channel depicted in Figure \ref{fig:channel_Z6}. We call this Channel 2.
\begin{figure}[h]
\centering
\includegraphics[scale=.7]{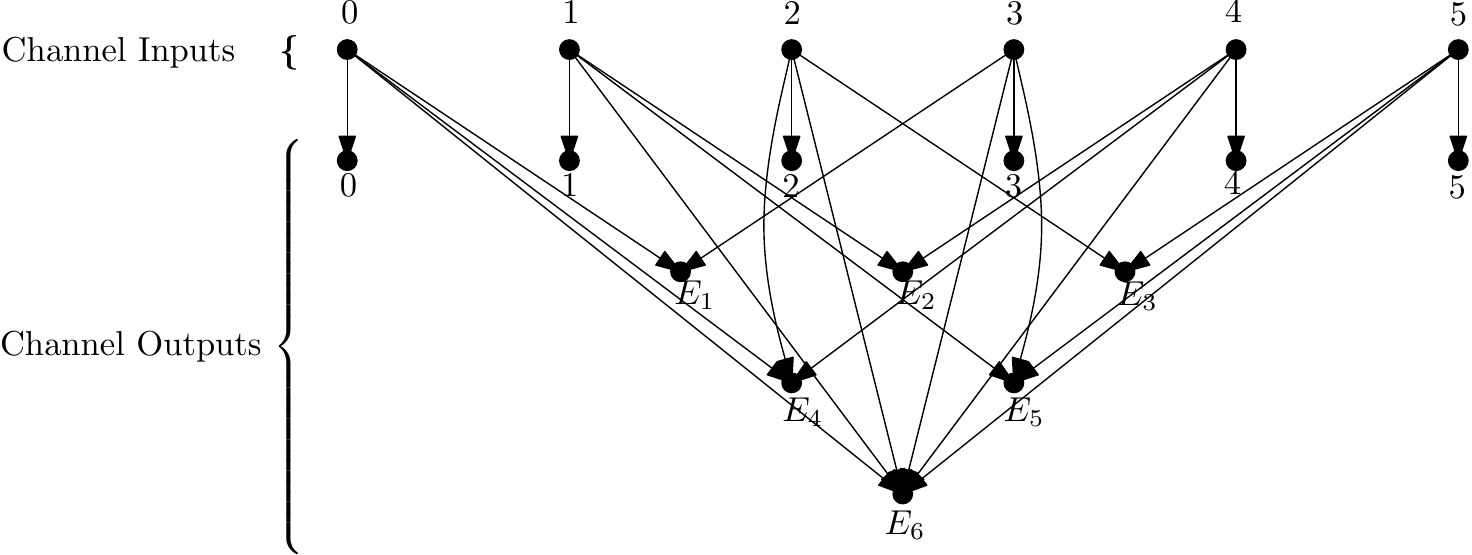}
\caption{\small Channel 2: A channel with a composite input alphabet size. For this channel, the process $I^n$ can be explicitly found for each $n$ and the multilevel polarization can be observed. $E_1$, $E_2$ and $E_3$ are erasures corresponding to cosets of the subgroup $\{0,3\}$ and $E_4$ and $E_5$ are erasures corresponding to cosets of the subgroup $\{0,2,4\}$. The lines connected to outputs $E_1,E_2$ and $E_3$ correspond to a conditional probability of $\gamma$, the lines connected to outputs $E_4$ and $E_5$ correspond to a conditional probability of $\epsilon$, the lines connected to the output $E_6$ correspond to a conditional probability of $\lambda$, and the lines connected to outputs $0,1,2,3,4$ and $5$ correspond to a conditional probability of $1-\gamma-\epsilon-\lambda$. The parameters $\gamma,\epsilon,\lambda$ take values from $[0,1]$ such that $\gamma+\epsilon+\lambda\le 1$.}
\label{fig:channel_Z6}
\end{figure}
It turns out that given a sequence of bits $b_1b_2\cdots b_n$, the transformed channel $W^{b_1b_2\cdots b_n}$ is (equivalent to) a channel of the same type as Channel $2$ but with possibly different parameters $\epsilon, \lambda$ and $\gamma$. At each step $n$, the corresponding parameters can be found using the following recursion:
Define $\epsilon_0=\epsilon$, $\lambda_0=\lambda$ and $\gamma_0=\gamma$. For $i=1,\cdots,n$,
\begin{itemize}
\item If $b_i=1$, let
\begin{align}\label{eqn:recursion3}
\left\{\begin{array}{l}
\gamma_i=\gamma_{i-1}^2+2\gamma_{i-1}\lambda_{i-1}\\
\epsilon_i=\epsilon_{i-1}^2+2\epsilon_{i-1}\lambda_{i-1}\\
\lambda_i=\lambda_{i-1}^2
\end{array}\right.
\end{align}
\item If $b_i=0$, let
\begin{align}\label{eqn:recursion4}
\left\{\begin{array}{l}
\gamma_i=2\gamma_{i-1}-\left(\gamma_{i-1}^2+2\gamma_{i-1}\lambda_{i-1}\right)\\
\epsilon_i=2\epsilon_{i-1}-\left(\epsilon_{i-1}^2+2\epsilon_{i-1}\lambda_{i-1}\right)\\
\lambda_i=2\lambda_{i-1}-\left(\lambda_{i-1}^2\right)
\end{array}\right.
\end{align}
\end{itemize}
Then we have
\begin{align*}
I(W^{b_1b_2\cdots b_n})=\log_2 6-\gamma_n\log_2 2-\epsilon_n\log_2 3-\lambda_n\log_2 6
\end{align*}
The proof of the recursion formulas for Channel $2$ is similar to that of Channel $1$ and is omitted. The fixed points of the functions corresponding to Equations \eqref{eqn:recursion3} and \eqref{eqn:recursion4} are given by $(0,0,0)$, $(1,0,0)$, $(0,1,0)$, $(1,1,0)$, $(0,0,1)$, $(-1,0,1)$, $(0,-1,1)$ and $(-1,-1,1)$, out of which $(0,0,0)$, $(1,0,0)$, $(0,1,0)$ and $(0,0,1)$ are admissible. Note that $(0,0,0)$ corresponds to a perfect channel with a capacity of $\log_2 6$ bits per channel use, $(1,0,0)$ corresponds to a partially perfect channel which can perfectly send the index of the coset of the subgroup $\{0,3\}$ to which the input belongs and has a capacity of $\log_2 3$ bits per channel use, $(0,1,0)$ corresponds to a partially perfect channel which can perfectly send the index of the coset of the subgroup $\{0,2,4\}$ to which the input belongs and has a capacity of $\log_2 2$ bits per channel use, and $(0,0,1)$ corresponds to a useless channel. This suggests that in the limit, the transformed channels converge to one of these four extreme cases. This can be confirmed using the recursion formulas for this channel as depicted in Figures \ref{fig:polarization_Z6_1} and \ref{fig:polarization_Z6_2}. With encoding and decoding rules similar to those of Channel 1, we can show that polar codes achieve the symmetric capacity of this channel.
\begin{figure}[!h]
\begin{minipage}[b]{0.5\linewidth}
\centering
\includegraphics[scale=.3]{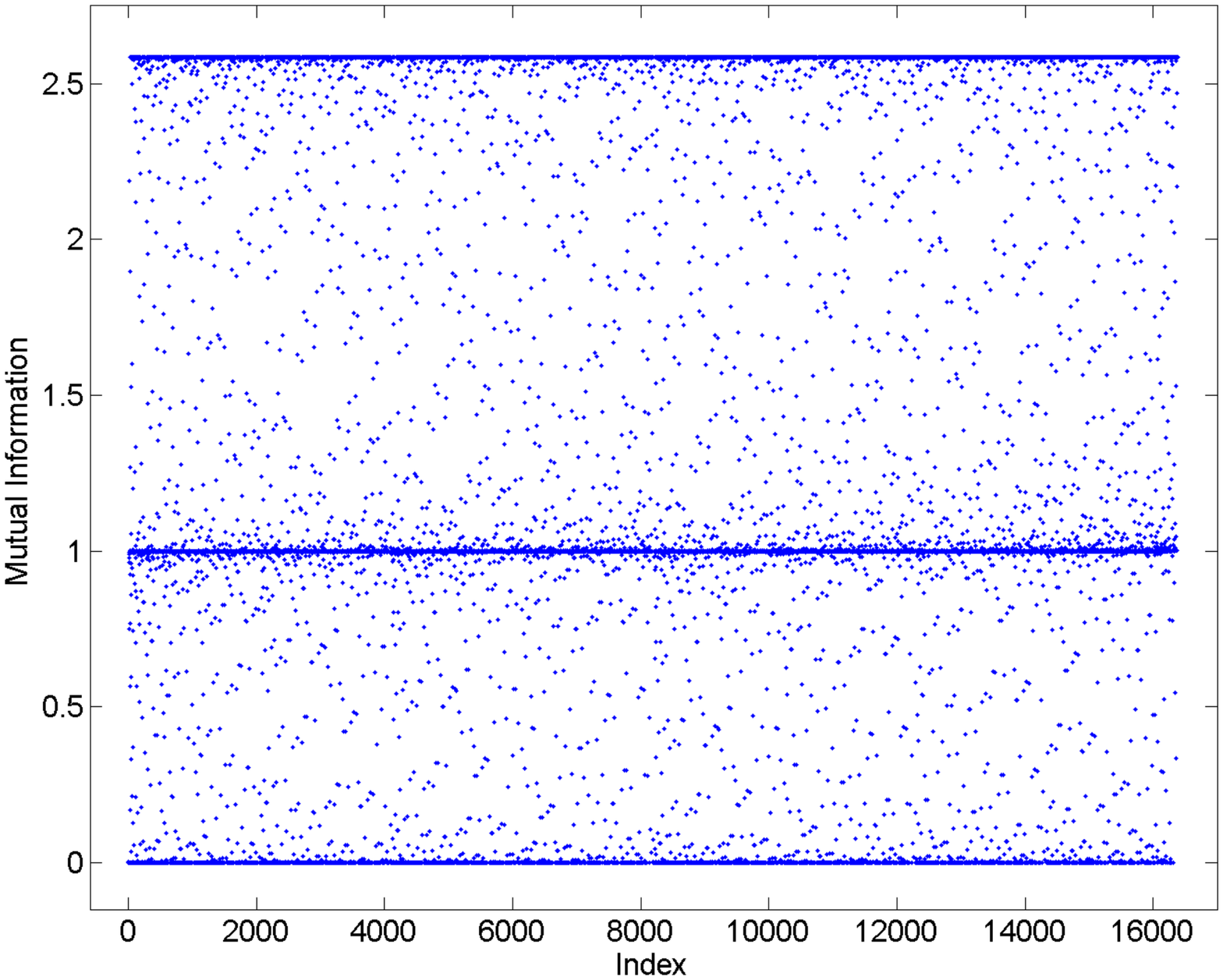}
\caption{\small Polarization of Channel 2 with parameters $\gamma=0, \epsilon=0.4, \lambda=0.2$. The middle line represents the subgroup $\{0,2,4\}$ of $\mathds{Z}_6$.}
\label{fig:polarization_Z6_1}
\end{minipage}
\hspace{0.5cm}
\begin{minipage}[b]{0.5\linewidth}
\centering
\includegraphics[scale=.3]{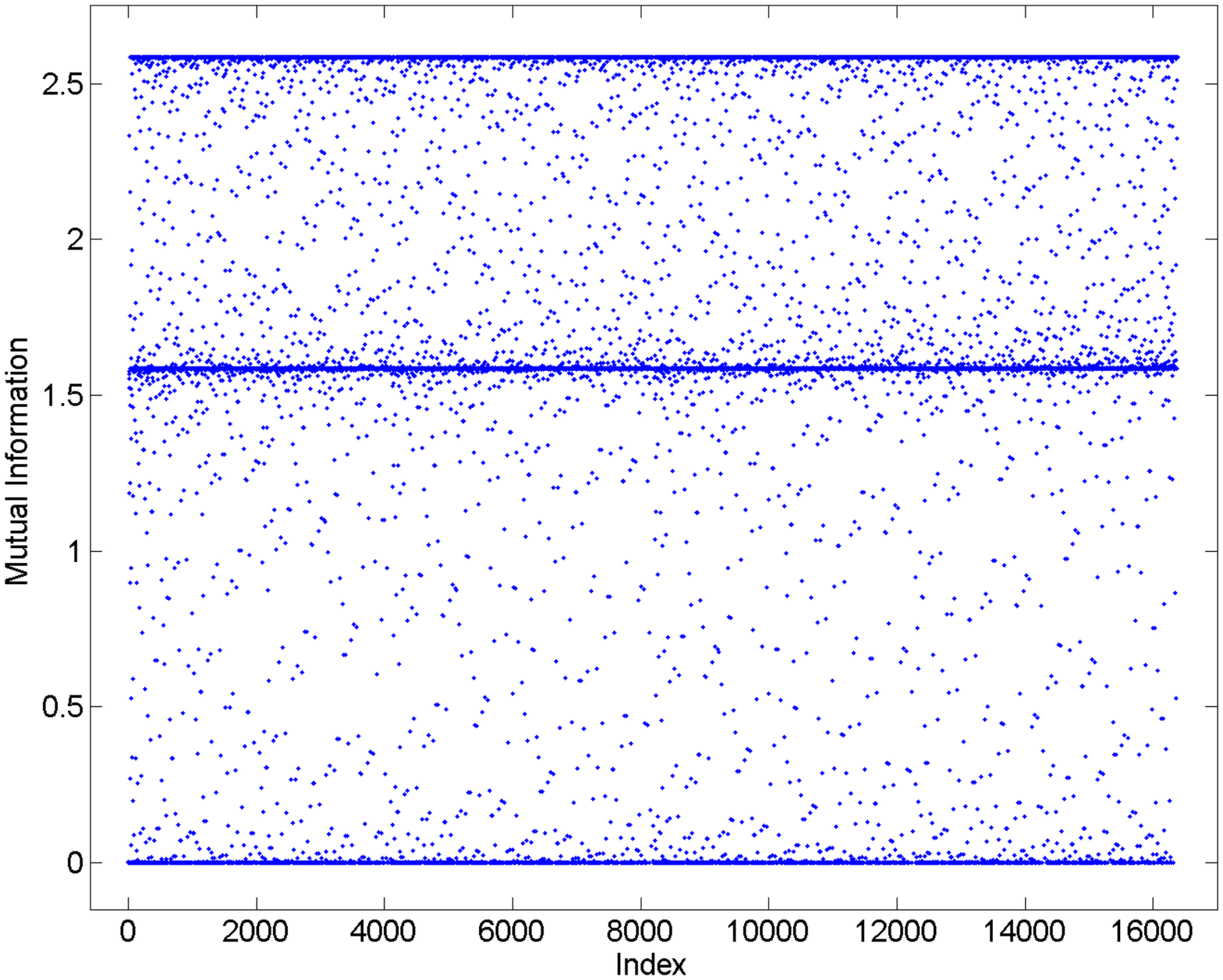}
\caption{\small Polarization of Channel 2 with parameters $\gamma=0.4, \epsilon=0, \lambda=0.2$. The middle line represents the subgroup $\{0,3\}$ of $\mathds{Z}_6$.}
\label{fig:polarization_Z6_2}
\end{minipage}
\end{figure}

In the next section, we show that polar codes achieve the symmetric capacity of channels with input alphabet size equal to a power of a prime.
\section{Polar Codes Over Channels with input $\mathds{Z}_{p^r}$}\label{rings}
In this section, we consider channels of input alphabet size $q=p^r$ for some prime number $p$ and a positive integer $r$. In this case, the input alphabet of the channel can be considered as a {ring} with addition and multiplication modulo $p^r$. We prove the achievability of the symmetric capacity of these channels using polar codes and later in Section \ref{section:abelian} we will generalize this result to channels of arbitrary input alphabet sizes and arbitrary group operations. We note that $O(\epsilon)$ functions used in this paper do not depend on the size of the channel output alphabet.

\subsection{$\mathds{Z}_{p^r}$ Rings}
Let $\G=\mathds{Z}_{p^r}=\{0,1,2,\cdots,p^r-1\}$ with addition and multiplication modulo $p^r$ be the input alphabet of the channel, where $p$ is a prime and $r$ is an integer. For $t=0,1,\cdots,r$, define the {subgroups} $H_t$ of $\G$ as the set:
\begin{align*}
H_t=p^t\G=\{0,p^t,2p^t,\cdots,(p^{r-t}-1)p^t\}
\end{align*}
and for $t=0,1,\cdots,r$, define the subsets $K_t$ of $\G$ as $K_t=H_t\backslash H_{t+1}$; i.e. $K_t$ is defined as the set of elements of $\G$ which are a multiple of $p^t$ but are not a multiple of $p^{t+1}$. Note that $K_0$ is the set of all invertible elements of $\G$ and $K_r=\{0\}$. One can sort the sets $K_0>K_1>\cdots>K_r$ in a decreasing order of ``invertibility'' of its elements. Let $T_t$ be a {transversal} of $H_t$ in $\G$; i.e. a subset of $\G$ containing one and only one element from each coset of $H_t$ in $\G$. One valid choice for $T_t$ is $\{0,1,\cdots,p^t-1\}$. Note that given $H_t$ and $T_t$, each element $g$ of $\G$ can be represented uniquely as a sum $g=\hat{g}+\tilde{g}$ where $\hat{g}\in T_t$ and $\tilde{g}\in H_t$.
\subsection{Recursive Channel Transformation}
\subsubsection{The Basic Channel Transforms}
It has been shown in \cite{arikan_polar} that the error probability of polar codes over binary input channels is upper bounded by the sum of Bhattacharyya parameters of certain channels defined by a recursive channel transformation. The same set of synthesized channels appear for polar codes over channels with arbitrary input alphabet sizes. The channel transformations are given by:
\begin{align}
&\label{eqn:channel_transform1} W^-(y_1,y_2|u_1)=\sum_{u_2^\prime\in \G}\frac{1}{q}W(y_1|u_1+u_2^\prime)W(y_2|u_2^\prime)\\
&\label{eqn:channel_transform2} W^+(y_1,y_2,u_1|u_2)=\frac{1}{q}W(y_1|u_1+u_2)W(y_2|u_2)
\end{align}
for $y_1,y_2\in\mathcal{Y}$ and $u_1,u_2\in \G$. Repeating these operations $n$ times recursively, we obtain $N=2^n$ channels $W_N^{(1)},\cdots,W_N^{(N)}$. For $i=1,\cdots,N$, these channels are given by:
\begin{align*}
W_N^{(i)}(y_1^N,u_1^{i-1}|u_i)=\sum_{u_{i+1}^N\in \G^{N-i}}\frac{1}{q^{N-1}}W^N(y_1^N|u_1^N G_N)
\end{align*}
where $G_N$ is the generator matrix for polar codes.\\
For the case of binary input channels, it has been shown in \cite{arikan_polar} that as $N\rightarrow \infty$, these channels polarize in the sense that their Bhattacharyya parameters gets either close to zero (perfect channels) or close to one (useless channels). In the next part, we show that in general, when the input alphabet is a prime power, polarization happens in multiple levels so that as $N\rightarrow \infty$ channels get useless, perfect or ``partially perfect''.\\
For an integer $n$, let $J_n$ be a uniform random variable over the set $\{1,2,\cdots,N=2^n\}$ and define the random variable $I^n(W)$ as
\begin{align}\label{eqn:Iprocess}
I^n(W)=I(X;Y)
\end{align}
where $X$ and $Y$ are the input and output of $W_N^{(J_n)}$ respectively and $X$ is uniformly distributed. It has been shown in \cite{sasoglu_polar_q} that the process $I^0,I^1,I^2,\cdots$ is a martingale; hence $\mathds{E}\{I^n\}=I^0$. For an integer $n$ and for $d\in\G$, define the random variable $Z_d^n(W)=Z_d(W_N^{(J_n)})$ where for a channel $(\G,\mathcal{Y},W)$,
\begin{align}\label{eqn:Zd}
Z_d(W)=\frac{1}{q}\sum_{x\in \G}\sum_{y\in \mathcal{Y}}\sqrt{W(y|x)W(y|x+d)}=\frac{1}{q}\sum_{x\in \G}Z(W_{\{x,x+d\}})
\end{align}
This quantity has been defined in \cite{sasoglu_polar_q}. Other than the processes $I^n(W)$ and $Z_d^n(W)$, in the proof of polarization, we need another set of processes $I^n_H(W)$ for $H\le \G$ which we define in the following. Let $H$ be an arbitrary subgroup of $\G$. Note that any uniform random variable defined over $\G$ can be decomposed into two uniform and independent random variables $\hat{X}$ and $\tilde{X}$ where $\hat{X}$ takes values from the transversal $T$ of $H$ and $\tilde{X}$ takes values from $H$. For an integer $n$, define the random variable $I^n_H(W)$ as
\begin{align}\label{eqn:Iprocess2}
I^n_H(W)=I(X;Y|\hat{X})=I(\tilde{X};Y|\hat{X})
\end{align}
where $X$ and $Y$ are the input and output of $W_N^{(J_n)}$ respectively. Next lemma shows that $I^n_H(W)$ is a super-martingale.

\begin{lemma}\label{lemma:supermartingale}
For an arbitrary group $\G$ and for any subgroup $H$ of $\G$, the random process $I^n_H(W)$ defined above is a super-martingale.
\end{lemma}
\begin{proof}
Define the channels $W^{-}$ and $W^{+}$ as in (\ref{eqn:channel_transform1}) and (\ref{eqn:channel_transform2}). Define the random variables $U_1$, $U_2$, $X_1$, $X_2$, $Y_1$ and $Y_2$ where $U_1$ and $U_2$ are uniformly distributed over $\G$, $X_1=U_1+U_2$ where addition is the group operation, $X_2=U_2$ and $Y_1$ (respectively $Y_2$) is the channel output when the input is $X_1$ (respectively $X_2$). Decompose the random variable  $U_1$ into two uniform and independent random variables $\hat{U}_1$ and $\tilde{U}_1$ where $\hat{U}_1$ takes values from the transversal $T$ of $H$ and $\tilde{U}_1$ takes values from $H$. Similarly define, $\hat{U}_2,\hat{X}_1,\hat{X}_2$ and $\tilde{U}_2,\tilde{X}_1,\tilde{X}_2$. We need to show that
\begin{align*}
I(\tilde{U}_1;Y_1Y_2|\hat{U}_1)+I(\tilde{U}_2;Y_1Y_2U_1|\hat{U}_2)\le 2I(\tilde{X}_1;Y_1|\hat{X}_1)
\end{align*}
Note that since $I^n$ is a martingale and $I(\tilde{X}_1;Y_1|\hat{X}_1)=I(X_1;Y_1)-I(\hat{X}_1;Y_1)$, it suffices to show
\begin{align*}
I(\hat{U}_1;Y_1Y_2)+I(\hat{U}_2;Y_1Y_2U_1)\ge 2I(\hat{X}_1;Y_1)
\end{align*}
We have
\begin{align*}
I(\hat{U}_2;Y_1Y_2U_1)&=I(\hat{U}_2;Y_1Y_2\hat{U}_1\tilde{U}_1)\\
&=I(\hat{U}_2;Y_1Y_2\hat{U}_1)+I(\hat{U}_2;\tilde{U}_1|Y_1Y_2\hat{U}_1)\\
&\ge I(\hat{U}_2;Y_1Y_2\hat{U}_1)
\end{align*}
Hence,
\begin{align*}
I(\hat{U}_1;Y_1Y_2)+I(\hat{U}_2;Y_1Y_2U_1)&\ge I(\hat{U}_1;Y_1Y_2)+ I(\hat{U}_2;Y_1Y_2\hat{U}_1)\\
&= I(\hat{U}_1 \hat{U}_2;Y_1Y_2)\\
&\stackrel{(a)}{=} I(\hat{X}_1 \hat{X}_2;Y_1Y_2)=2I(\hat{X_1};Y_1)
\end{align*}
where $(a)$ follows since $\hat{U}_1$ and $\hat{U}_2$ are recoverable from $\hat{X}_1$ and $\hat{X}_2$. To see this, let $U_1'$ and $U_2'$ take values form $\G$ and let $X_1'=U_1'+U_2'$ and $X_2'=U_2'$. We need to show that if $X_1'$ is in the same coset of $H$ as $X_1$ (i.e. if $X_1'-X_1\in H$ or equivalently $\hat{X}'_1=\hat{X}_1$) and $X_2'$ is in the same coset of $H$ as $X_2$ (i.e. if $X_2'-X_2\in H$ or equivalently $\hat{X}'_2=\hat{X}_2$), then $U_1'$ is in the same coset of $H$ as $U_1$ (i.e. $U_1'-U_1\in H$ or equivalently $\hat{U}'_1=\hat{U}_1$) and $U_2'$ is in the same coset of $H$ as $U_2$ (i.e. $U_2'-U_2\in H$ or equivalently $\hat{U}'_2=\hat{U}_2$). Note that $X_2'-X_2\in H$ implies $U_2'-U_2\in H$ and  $X_1'-X_1\in H$ implies $U_1'+U_2'-U_1-U_2\in H$. Since $U_2'-U_2\in H$ (and hence $U_2-U_2'\in H$), it follows that $U_1'-U_1\in H+U_2-U_2'=H$. This concludes the lemma.
\end{proof}
\subsubsection{Asymptotic Behavior of Synthesized Channels}
We restate Lemma 2 of \cite{sasoglu_polar_q} with a slight generalization:

\begin{lemma}\label{lemma:super_lemma2}
Suppose $B_n$, $n\in \mathds{Z}^+$ is a $\{-,+\}$-valued process with $P(B_n=-)=P(B_n=+)=\frac{1}{2}$. Suppose $I_n$ and $T_n$ are two processes adapted to the process $B_n$ satisfying the following conditions
\begin{enumerate}
\item $I_n$ takes values in the interval $[0,1]$.
\item $I_n$ converges almost surely to a random variable $I_\infty$.%is a super-martingale with respect to $\sigma(B_0,\cdots,B_n)$, the $\sigma$ algebra generated by $B_0,\cdots,B_n$.
\item $T_n$ takes values in the interval $[0,1]$.
\item $T_{n+1}=T_n^2$ when $B_{n+1}=+$.
\item If $T_n<\epsilon$ for all $n$, then $I_n>1-O(\epsilon)$ for all $n$, in the sense that there exists a function $f$ which is $O(\epsilon)$ and $T_n<\epsilon\Rightarrow I_n>1-f(\epsilon)$ for all $n$.
\item If $T_n>1-\epsilon$ for all $n$, then $I_n<O(\epsilon)$ for all $n$, in the sense that there exists a function $g$ which is $O(\epsilon)$ and $T_n>1-\epsilon\Rightarrow I_n<g(\epsilon)$ for all $n$.
\end{enumerate}
Then $I_\infty=\lim_{n\rightarrow \infty} I_n$ and $T_\infty=\lim_{n\rightarrow \infty} T_n$ both exist with probability $1$ and take values in $\{0,1\}$.
\end{lemma}
\begin{proof}
The proof follows from Lemma 2 of \cite{sasoglu_polar_q}. A sufficient condition for $I_n$ to converge is when $I_n$ is a bounded super-martingale. Note that condition $(i\&t.1)$ of Lemma 2 of \cite{sasoglu_polar_q} can be recovered from the last two conditions of this lemma. We use this notation to be consistent throughout the paper. To see this, note that $(5)$ and $(6)$ imply that there exist functions $f(\cdot),g(\cdot):\mathds{R}\rightarrow \mathds{R}$ such that $\lim_{\delta\downarrow 0} f(\delta)=0$ and $\lim_{\delta\downarrow 0} g(\delta)=0$ and that $T_n<\delta$ implies $I_n>1-f(\delta)$ and $T_n>1-\delta$ implies $I_n<g(\delta)$. For an arbitrary $\epsilon>0$, since the limit of both functions at zero is zero, let $\delta>0$ be such that $f(\delta)<\epsilon$ and $g(\delta)<\epsilon$. For this choice of $\delta$ we have
\begin{align*}
&T_n<\delta \Rightarrow I_n>1-f(\delta)>1-\epsilon\\
&T_n>1-\delta \Rightarrow I_n<g(\delta)<\epsilon
\end{align*}
Hence for any (sufficiently small) $\epsilon>0$, there exists a $\delta>0$ such that $T_n<\delta$ implies $I_n>1-\epsilon$ and $T_n>1-\delta$ implies $I_n<\epsilon$. Equivalently, for any $\epsilon>0$, there exists a $\delta>0$ such that $\epsilon\le I_n\le 1-\epsilon$ implies $\delta\le T_n\le 1-\delta$.
\end{proof}
In the next lemma, we show that for any $d\in\G$, the random process $Z_d^n$ converges to a Bernoulli random variable.

\begin{lemma}\label{Zd_0_1_valued}
For all $d\in \G$, $Z^n_d(W)$ converges to a $\{0,1\}$-valued random variable $Z^\infty_d(W)$ as $n$ grows. Moreover, if $\tilde{d}\in\G$ is such that $\langle \tilde{d}\rangle=\langle d\rangle$ then $Z^\infty_{\tilde{d}}(W)=Z^\infty_d(W)$ almost surely; i.e. the random processes $Z^n_{\tilde{d}}(W)$ and $Z^n_d(W)$ converge to the same random variable.
\end{lemma}

\begin{proof}%\ref{Zd_0_1_valued}
This lemma has been proved in \cite[Theorem 1]{sasoglu_polar_q} for $d=\arg\max_{a\ne 0} Z_a(W)$ when the underlying group is a field. The proof for an arbitrary $d$ and an arbitrary group is given in the following. Let $H=\langle d \rangle$ be the subgroup of $\G$ generated by $d$ and let $M$ be a maximal subgroup of $H$. Then the proof provided in \cite{sasoglu_polar_q} suffices for this lemma if we consider the quotient group $H\slash M$ which is of prime order. We will elaborate on this in the following: Let
%\begin{align*}
%d^\prime=\arg\max_{\substack{a\in H\\a\notin M \quad \forall M<H}} Z_a(W)
%\end{align*}
\begin{align}\label{eqn:d_maximal_sg}
d^\prime=\arg\max_{\substack{a\in H\\a\notin M}} Z_a(W)
\end{align}
In Lemma \ref{lemma:super_lemma2}, let $I^n$ (Here we use the notation $I^n$ instead of $I_n$ for notational convenience) be equal to the process $I^n_H(W)-I^n_M(W)$ where $I^n_H(W)$ and $I^n_M(W)$ are defined by Equation \eqref{eqn:Iprocess2} and let $T_n$ be equal to the process $Z_{d^\prime}^n(W)$ defined in (\ref{eqn:Zd}). We claim that $I^n$ and $T_n$ satisfy the conditions of Lemma \ref{lemma:super_lemma2}. The proof is given in the following:\\ \\
%Then the condition $(1)$ is satisfied by a careful choice of the base of the $\log$ function (base $|H|/|M|$, the ratio of the cardinalities of the subgroups). $(2)$ follows from Lemma \ref{lemma:super_lemma2} and the fact that both $I^n_H(W)$ and $I^n_M(W)$ are bounded. $(3)$ is trivial. $(4)$ has been shown in \cite{sasoglu_polar_q} and $(5)$ and $(6)$ follow from Lemmas \ref{lemma:Zd1_22} and \ref{lemma:Zd0} and continuity of mutual information (see Appendix \ref{section:lemmas}). Since this proof is based on knowledge of group theory, we give more explanations in the rest of the proof:\\ \\
Note that in the case of $\mathds{Z}_p$ fields, the only maximal subgroup of the group is the trivial subgroup $\{0\}$. Hence, (\ref{eqn:d_maximal_sg}) can be viewed as a straightforward generalization of the the definition made in \cite{sasoglu_polar_q}. Let $M$ be a maximal subgroup of $H=\langle d\rangle$. Recall that a uniform random variable $X$ over $\G$ can be decomposed into two uniform and independent random variables $\tilde{X}$ taking values from $H$ and $\hat{X}$ taking values from the transversal of $H$ in $\G$. Similarly, the uniform random variable $\tilde{X}$ over $H$ can be decomposed into two uniform and independent random variables $\tilde{\tilde{X}}$ taking values from $M\le H$ and $\hat{\tilde{X}}$ taking values from the transversal of $M$ in $H$. Using the chain rule we have:
\begin{align*}
I(\tilde{X};Y|\hat{X})&=I(\tilde{\tilde{X}}\hat{\tilde{X}};Y|\hat{X})\\
%&=H(\tilde{\tilde{X}}\hat{\tilde{X}}|\hat{X})-H(\tilde{\tilde{X}}\hat{\tilde{X}}|\hat{X}Y)\\
%&=H(\hat{\tilde{X}}|\hat{X})+H(\tilde{\tilde{X}}|\hat{X}\hat{\tilde{X}})- H(\hat{\tilde{X}}|\hat{X}Y)-H(\tilde{\tilde{X}}|\hat{X}\hat{\tilde{X}}Y)\\
%&=\left(H(\hat{\tilde{X}}|\hat{X})- H(\hat{\tilde{X}}|\hat{X}Y)\right)+ \left(H(\tilde{\tilde{X}}|\hat{X}\hat{\tilde{X}})- H(\tilde{\tilde{X}}|\hat{X}\hat{\tilde{X}}Y)\right)\\
&=I(\hat{\tilde{X}};Y|\hat{X})+I(\tilde{\tilde{X}};Y|\hat{X}\hat{\tilde{X}})
\end{align*}
Note that $\tilde{\tilde{X}}\in M$ and $(\hat{X},\hat{\tilde{X}})$ indicate the coset of $M$ in $\G$ to which $X$ belongs. Therefore, the equation above implies that for each $n$, $I^n_H(W)-I^n_M(W)=I(\hat{\tilde{X}};Y|\hat{X})$ where $X$ and $Y$ are the input and the output of the channel $W_N^{(J_n)}$. Since $\hat{\tilde{X}}$ can at most take $\frac{|H|}{|M|}$ values, by choosing the base of the $\log$ function to be equal to $\frac{|H|}{|M|}$ condition $(1)$ of Lemma \ref{lemma:super_lemma2} satisfies.\\ \\
We have shown in Lemma \ref{lemma:supermartingale} that both processes $I^n_H(W)$ and $I^n_M(W)$ are super-martingales and hence both converge almost surely. This means that the vector valued random process $(I^n_H(W),I^n_M(W))$ converges almost surely (refer to Proposition 5.25 of \cite{Karr_Probability}).
% We need to show that
%\begin{align*}
%I(\hat{\tilde{U}_1};Y_1Y_2|\hat{U}_1)+I(\hat{\tilde{U}}_2;Y_1Y_2U_1|\hat{U}_2)\ge 2 I(\hat{\tilde{X}}_1;Y_1|\hat{X}_1)
%\end{align*}
%Note that
%\begin{align*}
%I(\hat{\tilde{U}}_2;Y_1Y_2U_1|\hat{U}_2)&= H(\hat{\tilde{U}}_2|\hat{U}_2)-H(\hat{\tilde{U}}_2|\hat{U}_2 Y_1Y_2U_1)\\
%&\ge H(\hat{\tilde{U}}_2|\hat{U}_2)-H(\hat{\tilde{U}}_2|\hat{U}_2 Y_1Y_2\hat{\tilde{U}_1})\\
%&= H(\hat{\tilde{U}}_2)-H(\hat{\tilde{U}}_2|\hat{U}_2 Y_1Y_2\hat{\tilde{U}}_1)
%\end{align*}
%and
%\begin{align*}
%I(\hat{\tilde{U}}_1;Y_1Y_2|\hat{U}_1)= H(\hat{\tilde{U}}_1|\hat{U}_1)-H(\hat{\tilde{U}}_1|\hat{U}_1 Y_1Y_2)\\
%&= H(\hat{\tilde{U}}_1)-H(\hat{\tilde{U}}_1|\hat{U}_1 Y_1Y_2)
%\end{align*}
%Therefore
%\begin{align*}
%I(\hat{\tilde{U}}_1;Y_1Y_2|\hat{U}_1)+I(\hat{\tilde{U}}_2;Y_1Y_2U_1|\hat{U}_2)&\ge H(\hat{\tilde{U}}_1\hat{\tilde{U}}_2)-
%\end{align*}
Hence condition $(2)$ is satisfied.\\ \\
Condition $(3)$ trivially holds and condition $(4)$ is shown to be satisfied in the proof of Theorem 1 of \cite{sasoglu_polar_q}.\\ \\
To show (5), assume $Z_{d^\prime}^n(W)<\epsilon$. Let $T_H$ be a transversal of $H$ in $\G$ and let $T_M$ be a transversal of $M$ in $H$. Given $X\in t_H+H$ for some $t_H \in T_H$, the joint probability distribution of cosets of $M$ in $t_H+H$ and the channel output is given by:
\begin{align*}
\bar{p}(t_H+t_M+M,y)& \triangleq\sum_{m\in M}P(X=t_H+t_M+m,Y=y|X\in t_H+H)\\
&=\sum_{m\in M}\frac{P(X=t_H+t_M+m,Y=y)}{P(X\in t_H+H)}\\
&=\sum_{m\in M}\frac{P(X=t_H+t_M+m,Y=y)}{|H|/|\G|}\\
&=\frac{|\G|}{|H|}\sum_{m\in M}\frac{1}{|\G|}W(y|t_H+t_M+m)\\
&=\frac{1}{|H|}\sum_{m\in M}W(y|t_H+t_M+m)
\end{align*}
where $t_M$ takes values from $T_M$. The corresponding channel is defined as:
\begin{align}\label{eq:W_bar_def}
\nonumber \bar{W}(y|t_H+t_M+M)&=\frac{1}{P(X\in t_H+t_M+M|X\in t_H+H)}\frac{1}{|H|}\sum_{m\in M}W(y|t_H+t_M+m)\\
&=\frac{1}{|M|}\sum_{m\in M}W(y|t_H+t_M+m)
\end{align}
Note that the input of this channel takes values from the set $\{t_H+t_M+M|t_M\in T_M\}$ uniformly and the size of the input alphabet is $\frac{|H|}{|M|}\triangleq \bar{q}$ which is a prime (since $M$ is maximal in $H$). Furthermore, by definition $I(\bar{W})=I(\hat{\tilde{X}};Y|\hat{X}=t_H)$. It is shown in Appendix \ref{section:Upper_Bound_ZWbar} that $Z_{d^\prime}(W)<\epsilon$ implies $Z(\bar{W})<C\epsilon$ for some constant $C=\frac{|M|\cdot |H|\cdot |G|}{|H|-|M|}$. Therefore, \cite[Prop. 3]{sasoglu_polar_q} implies $I(\bar{W})=\log \frac{|H|}{|M|}-O(\epsilon)$. This result is valid for all $t_H\in T_H$. Therefore
\begin{align*}
I_H(W)-I_M(W)&=\sum_{t_H\in T_H} P(\hat{X}=t_H) I(\hat{\tilde{X}};Y|\hat{X}=t_H)\\
&=\log \frac{|H|}{|M|}-O(\epsilon)
\end{align*}
%It follows that for all $\tilde{d}\in H\backslash M$, $Z_{\tilde{d}}^n(W)<\epsilon$ and hence from Lemma \ref{lemma:Zd0}, for all $x\in G$ and all $y\in \mathcal{Y}$ either $W(y|x)<O(\epsilon)$ or $W(y|x+\tilde{d})<O(\epsilon)$. This means if $x$ and $x^\prime$ are in different cosets of $M$ in $H$ then either $W(y|x)<O(\epsilon)$ or $W(y|x')<O(\epsilon)$. Note that $I^n_H(W)-I^n_M(W)=I(\hat{\tilde{X}};Y|\hat{X})$ where $\hat{\tilde{X}}$ indicates the coset of $M$ in $H$ to which $\tilde{X}$ belongs. This and continuity of mutual information imply $I(\hat{\tilde{X}};Y|\hat{X})=\log \frac{|H|}{|M|}-O(\epsilon)$. This shows that $(5)$ is satisfied.\\ \\

To show condition (6), assume that $Z_{d^\prime}^n(W)>1-\epsilon$. For the channel $\bar{W}$ defined as above, it is shown in Appendix \ref{section:Lower_Bound_ZWbar1} (An alternate proof for the $\mathds{Z}_{p^r}$ case can be found in Appendix \ref{section:Alt_Lower_Bound_ZWbar}) that $Z_{d^\prime}(W)>1-\epsilon$ implies $Z_{d'+t_H+M}(\bar{W})>1-\frac{2q(2\epsilon-\epsilon^2)}{\bar{q}|M|}=1-O(\epsilon)$. Since the input alphabet of the channel $\bar{W}$ has a prime size and $d'\in H\backslash M$, we can use \cite[Lemma 4]{sasoglu_polar_q} to conclude that $Z(\bar{W})>1-\frac{2q\bar{q}^2(2\epsilon-\epsilon^2)}{|M|}=1-O(\epsilon)$. Now we use \cite[Prop. 3]{sasoglu_polar_q} to conclude $I(\bar{W})<O(\epsilon)$. This implies:
\begin{align*}
I_H(W)-I_M(W)&=\sum_{t_H\in T_H} P(\hat{X}=t_H) I(\hat{\tilde{X}};Y|\hat{X}=t_H)\\
&<O(\epsilon)
\end{align*}

So far, we have shown that for any $d\in G$, for $H=\langle d\rangle$ and $d'$ defined as in \eqref{eqn:d_maximal_sg}, the random variable $Z_{d'}^n(W)$ converges to a Bernoulli random variable. Note that so far the proof is general and applies to arbitrary groups as well. We will use this part of the proof later in Section \ref{section:abelian}. Next, we show that when $G=\mathds{Z}_{p^r}$, for any $\tilde{d}\in H\backslash M$ (including $d$ itself), $Z_{\tilde{d}}^n(W)$ converges to a Bernoulli random variable. Moreover, all such $\tilde{d}$'s converge to the same random variable. To see this, note that if $Z_{d'}^n<\epsilon$, it follows that $Z_{\tilde{d}}^n<\epsilon$ for all $\tilde{d}\in H\backslash M$ and if $Z_{d'}^n>1-\epsilon$ we show that for all $\tilde{d}\in \langle d'\rangle=H$, $Z_{\tilde{d}}^n>1-O(\epsilon)$. For any $\tilde{d}\in H=\langle d\rangle$ we can write $\tilde{d}=id'$ for some integer $i$. The condition $Z_{d'}> 1-\epsilon$ implies $1-Z(W_{\{x,x+d'\}})\le q\epsilon$ for all $x\in\G$. It has been shown in the proof of \cite[Lemma 4]{sasoglu_polar_q} that
\begin{align*}
\sqrt{1-Z(W_{\{x,x+2d'\}})}\le \sqrt{1-Z(W_{\{x,x+d'\}})}+\sqrt{1-Z(W_{\{x+d,x+2d'\}})}\le 2\sqrt{q\epsilon}
\end{align*}
Repeated application of this inequality for $i$ times yields $\sqrt{1-Z(W_{\{x,x+\tilde{d}\}})}\le i\sqrt{q\epsilon}\le q\sqrt{q\epsilon}$ or equivalently $Z(W_{\{x,x+\tilde{d}\}})\ge 1-q^3\epsilon$. It then follows that $Z_{\tilde{d}}\ge 1-q^3\epsilon$. Note that when $\G=\mathds{Z}_{p^r}$, $H\backslash M$ is the set of all elements $\tilde{d}$ such that $\langle \tilde{d}\rangle=\langle d\rangle$. This completes the proof of the lemma.
\end{proof}

The next lemma gives a sufficient condition for two processes $Z_d^n$ and $Z_{\tilde{d}}^n$ to converge to the same random variable. Recall that for $0\le t\le r-1$, $K_t=H_t\backslash H_{t+1}$.

\begin{lemma}\label{lemma:d_dprime_in_Kt}
If $d,\tilde{d}\in K_t$ for some $0\le t\le r-1$, then $Z_d^n$ and $Z_{\tilde{d}}^n$ converge to the same Bernoulli random variable.
\end{lemma}

\begin{proof}%(of Lemma \ref{lemma:d_dprime_in_Kt})
Note that $d,\tilde{d}\in K_t$ implies $\langle d\rangle=\langle \tilde{d}\rangle=H_t$. Therefore, Lemma \ref{Zd_0_1_valued} implies $Z_d^n$ and $Z_{\tilde{d}}^n$ converge to the same Bernoulli random variable.

%Note that $d\in K_t$ implies $\langle d\rangle=H_t$ and $M=H_{t+1}$ is a maximal subgroup of $H_t$. We stated in the above lemma that for all $a\in H_t\backslash M$, $Z_a^n(W)$ converge to the same random variable as $Z_d^n(W)$. This lemma follows since $\tilde{d}\in K_t=H\backslash M$.
%It has been shown that $Z_d^n$ and $Z_{\tilde{d}}^n$ both converge to Bernoulli random variables. It suffices to show that $Z_d^n> 1-\epsilon$ implies $Z_{\tilde{d}}^n \approx_{\epsilon} 1$ and $Z_{\tilde{d}}^n> 1-\epsilon$ implies $Z_{d}^n \approx_{\epsilon} 1$. First assume $Z_d^n> 1-\epsilon$. Lemma \ref{lemma:Zd1_22} implies that for all $y\in \mathcal{Y}$, if $x-x^\prime\in \langle d\rangle=\langle \tilde{d}\rangle$ then $W(y|x)\approx_{\epsilon} W(y|x^\prime)$. This and Lemma \ref{lemma:Zd1_3} (see Appendix \ref{section:lemmas}) imply that $Z_{\tilde{d}}^n\approx_{\epsilon} 1$. Similarly we can show that $Z_{\tilde{d}}^n> 1-\epsilon$ implies $Z_{d}^n \approx_{\epsilon} 1$.
\end{proof}

For $t=0,1,\cdots,r-1$, pick an arbitrary element $k_t\in K_t$. The lemma above suggests that we only need to study $Z_{k_t}$'s rather than all $Z_d$'s.

\begin{lemma}\label{lemma:k_tk_s}
If $Z_{k_t}> 1-\epsilon$ then $Z_{k_s}\approx_{\epsilon} 1$ for all $t\le s\le r-1$.
\end{lemma}

\begin{proof}
Note that $k_s\in\langle k_t \rangle$ and let $d=k_t$ and $k_s=id$ for some integer $i$. The condition $Z_{k_t}> 1-\epsilon$ implies $1-Z(W_{\{x,x+d\}})\le q\epsilon$ for all $x\in\G$. It has been shown in the proof of \cite[Lemma 4]{sasoglu_polar_q} that  for all $x\in\G$
\begin{align*}
\sqrt{1-Z(W_{\{x,x+2d\}})}\le 2\sqrt{q\epsilon}
\end{align*}
Repeated application of this inequality for $i$ times yields $\sqrt{1-Z(W_{\{x,x+k_s\}})}\le i\sqrt{q\epsilon}$  for all $x\in\G$. It follows that $Z_{k_s}\ge 1-O(\epsilon)$.
%Follows from Lemma \ref{lemma:Zd1_22} and Lemma \ref{lemma:Zd1_3} and the fact that $k_s\in \langle k_t\rangle$.
\end{proof}

This lemma implies that for the group $\G=\mathds{Z}_{p^r}$ all possible  asymptotic cases are:
\begin{itemize}
\item \textbf{Case 0:} $Z_{k_0}= 1,Z_{k_1}= 1,Z_{k_2}= 1,\cdots,Z_{k_{r-1}}= 1$
\item \textbf{Case 1:} $Z_{k_0}= 0,Z_{k_1}= 1,Z_{k_2}= 1,\cdots,Z_{k_{r-1}}= 1$
\item \textbf{Case 2:} $Z_{k_0}= 0,Z_{k_1}= 0,Z_{k_2}= 1,\cdots,Z_{k_{r-1}}= 1$\\
\vdots
\item \textbf{Case r:} $Z_{k_0}= 0,Z_{k_1}= 0,Z_{k_2}= 0,\cdots,Z_{k_{r-1}}= 0$,
\end{itemize}
where for $t=0,\cdots,r$, case $t$ happens with some probability $p_t$.\\
Next, we study the behavior of $I^n$ in each of these asymptotic cases.

\begin{lemma}\label{lemma:I_t}
For a channel $(\mathds{Z}_{p^r},\mathcal{Y},W)$ and  for $t=0,1,\cdots,r$, if $Z_{k_0}<\epsilon ,Z_{k_1}<\epsilon,\cdots,Z_{k_{t-1}}<\epsilon, Z_{k_t}>1-\epsilon,\cdots,Z_{k_{r-1}}>1-\epsilon$, then $t\log p -O(\epsilon)<I^0(W)< t\log p+O(\epsilon)$.
\end{lemma}

\begin{proof}
Note that for all $s=0,\cdots,r-1$, $M_s\triangleq\langle k_{s+1}\rangle$ is a maximal subgroup of $\langle k_{s}\rangle$. In the proof of Lemma \ref{Zd_0_1_valued}, if we let $d=k_0$ and $M_0=\langle k_{1}\rangle$, we get $I_{\G}(W)-I_{M_0}(W)=I(W)-I_{M_0}(W)\approx_\epsilon \log p$ (Here we take the base of the $\log$ function to be equal to $2$). Similarly, it follows that $I_{M_s}(W)-I_{M_{s+1}}(W)\approx_\epsilon \log p$ for all $0\le s\le t-1$. For $s\ge t$ we have, $I_{M_s}-I_{M_{s+1}}\approx_{\epsilon} 0$. Therefore,
\begin{align*}
I^0(W)=I_{\G}(W)&=\sum_{s=0}^{r-1}I_{M_s}(W)-I_{M_{s+1}}(W)\\
&=\sum_{s=0}^{t-1} I_{M_s}(W)-I_{M_{s+1}}(W)+\sum_{s=t}^{r-1} I_{M_s}(W)-I_{M_{s+1}}(W)\\
&\approx_{\epsilon} \sum_{s=0}^{t-1} \log p+\sum_{s=t}^{r-1} 0\\
&=t\log p
\end{align*}
%We first prove the statement for the limiting case: $Z_{k_0}=0 , Z_{k_1}=0, \cdots,Z_{k_{t-1}}=0, Z_{k_t}=1, \cdots, Z_{k_{r-1}}=1$. From Lemma \ref{lemma:Zd1_22}, we have $W(y|x)=W(y|\tilde{x})$ if $x-\tilde{x}\in H_t$ and from Lemma \ref{lemma:Zd0} we have $W(y|x)W(y|\tilde{x})=0$ if $x-\tilde{x}\in K_0\cup K_1\cup\cdots\cup K_{t-1}=\G\backslash H_t$. Therefore, for any $y\in\mathcal{Y}$ with positive probability $p_Y(y)$, $W(y|x)$ is uniform over a coset of $H_t$ and zero over all other cosets. i.e. for all $y\in\mathcal{Y}$, there exists a coset $C_t^y$ of $H_t$ such that $\frac{1}{q}W(y|x)=\frac{p_Y(y)}{|C_t^y|}=\frac{p_Y(y)}{p^{r-t}}$ for $y\in C_t^y$ and $W(y|x)=0$ otherwise. The mutual information is equal to
%\begin{align*}
%I^0(W)&=\sum_{x\in\G}\sum_{y\in\mathcal{Y}}\frac{1}{q}W(y|x)\log\frac{W(y|x)}{\sum_{\tilde{x} \in\mathcal{X}}\frac{1}{q}W(y|\tilde{x})}\\
%&=\sum_{y\in\mathcal{Y}} \sum_{x\in C_t^y} \frac{1}{q}W(y|x)\log\frac{W(y|x)}{\sum_{\tilde{x} \in C_t^y}\frac{1}{q}W(y|\tilde{x})}\\
%&=\sum_{y\in\mathcal{Y}} \sum_{x\in C_t^y} \frac{p_Y(y)}{p^{r-t}}\log\frac{\frac{p^r p_Y(y)}{p^{r-t}}}{\sum_{\tilde{x} \in C_t^y}\frac{p_Y(y)}{p^{r-t}}}\\
%&=\sum_{y\in\mathcal{Y}} p^{r-t}\frac{p_Y(y)}{p^{r-t}}\log\frac{\frac{p^r p_Y(y)}{p^{r-t}}}{p^{r-t} \frac{p_Y(y)}{p^{r-t}}}\\
%&=\sum_{y\in\mathcal{Y}} p_Y(y)\log p^t=t\log p
%\end{align*}
%The lemma is proved considering the continuity of the mutual information.
\end{proof}

We have shown that the process $I^n$ converges to the following $r+1$ valued discrete random variable: $I^{\infty}=t\log p$ with probability $p_t$ for $t=0,\cdots,r$.\\
For $t=0,\cdots,r$, define the random variable $Z^t(W_N^{(i)})=\sum_{d\notin H_t}Z_d(W_N^{(i)})$ and the random process $(Z^t)^{(n)}(W)=Z^t(W^{(J_n)}_N)$ where $J_n$ is a uniform random variable over $\{1,2,\cdots,N=2^n\}$. Note that $(Z^t)^{(n)}(W)$ converges to a random variable $(Z^t)^{(\infty)}(W)$ almost surely and $P\left((Z^t)^{(\infty)}=0\right)=\sum_{s=t}^r p_s$.

\subsubsection{Summary of Channel Transformation}
For the channel $(\mathds{Z}_{p^r},\mathcal{Y},W)$, consider the vector random process $\mathbf{V}^n=(Z_{k_0}^n,Z_{k_1}^n,\cdots,Z_{k_{r-1}}^n,I^n)$. We have seen in the previous section that each component of this vector random process converges almost surely. Proposition 5.25 of \cite{Karr_Probability} implies that the vector random process $\mathbf{V}^n$ also converges almost surely to a random vector $\mathbf{V}^{\infty}$. The random vector $\mathbf{V}^{\infty}$ is a discrete random variable defined as follows:
\begin{align*}
P\left(\mathbf{V}^{\infty}=(\overbrace{0,\cdots,0}^{t\mbox{ \scriptsize times}},\overbrace{1,\cdots,1}^{r-t\mbox{ \scriptsize times}},t\log p)\right)=p_t
\end{align*}
for $t=0,1,\cdots,t$ where $p_t$'s are some probabilities. This implies that for all $\epsilon>0$, there exists a number $N=N(\epsilon)=2^{n(\epsilon)}$ and a partition $\{A_0^\epsilon,A_1^\epsilon, \cdots, A_r^\epsilon\}$ of $\{1,\cdots,N\}$ such that for $t=0,\cdots,r$ and $i\in A_t^\epsilon$, $Z_{k_s}(W_N^{(i)})<O(\epsilon)$ if $0\le s<t$ and $Z_{k_s}(W_N^{(i)})>1-O(\epsilon)$ if $t\le s<r$. For $t=0,\cdots,r$ and $i\in A_t^\epsilon$, we have $I(W_N^{(i)})=t\log(p)+O(\epsilon)$ and $Z^t(W_N^{(i)})=O(\epsilon)$. Moreover, as $\epsilon\rightarrow 0$, $\frac{|A_t^\epsilon|}{N}\rightarrow p_t$ for some probabilities $p_0,\cdots,p_r$.\\ \\
In Appendix \ref{section:Rate_of_Polarization}, we show that for any $\beta<\frac{1}{2}$ and for $t=0,\cdots,r$,
\begin{align}\label{rate_of_pol}
\lim_{n\rightarrow \infty}P\left((Z^t)^{(n)}<2^{-2^{\beta n}}\right)&\ge P\left((Z^t)^{(\infty)}=0\right)\\
\nonumber &=\sum_{s=t}^r p_s
\end{align}
\begin{remark}\label{remark_rate_of_pol}
This observation implies the following stronger result: For all $\epsilon>0$, there exists a number $N=N(\epsilon)=2^{n(\epsilon)}$ and a partition $\{A_0^\epsilon,A_1^\epsilon, \cdots, A_r^\epsilon\}$ of $\{1,\cdots,N\}$ such that for $t=0,\cdots,r$ and $i\in A_t^\epsilon$, $I(W_N^{(i)})=t\log(p)+O(\epsilon)$ and $Z^t(W_N^{(i)})<2^{-2^{\beta n(\epsilon)}}$. Moreover, as $\epsilon\rightarrow 0$, $\frac{|A_t^\epsilon|}{N}\rightarrow p_t$ for some probabilities $p_0,\cdots,p_r$.
\end{remark}

\subsection{Encoding and Decoding}
In the original construction of polar codes, we fix the input symbols corresponding to useless channels and send information symbols over perfect channels. Here, since the channels do not polarize into two levels, the encoding is slightly different and we send ``some'' information bits over ``partially perfect'' channels. At the encoder, if $i\in A_t^\epsilon$ for some $t=0,\cdots,r$, the information symbol is chosen from the transversal $T_t$ arbitrarily and not from the whole set $\G$. As we will see later, the channel $W_N^{(i)}$ is perfect for symbols chosen from $T_t$ and perfect decoding is possible at the decoder. Let $\mathcal{X}_N^\epsilon=\bigoplus_{t=0}^r T_t^{A_t^\epsilon}$ be the set of all valid input sequences. For the sake of analysis, as in the binary case, the message $u_1^N$ is dithered with a uniformly distributed random vector $b_1^N\in\bigoplus_{t=0}^r H_t^{A_t^\epsilon}$ revealed to both the encoder and the decoder. A message $v_1^N\in\mathcal{X}_N^\epsilon$ is encoded to the vector $x_1^N=(v_1^N+b_1^N) G_N$. Note that $u_1^N=v_1^N+b_1^N$ is uniformly distributed over $\G^N$.\\
At the decoder, after observing the output vector $y_1^N$, for $t=0,\cdots,r$ and $i\in A_t^{\epsilon}$ , use the following decoding rule:
\begin{align*}
\hat{u}_i =f_i(y_1^N,\hat{u}_1^{i-1})=\argmax_{g\in b_i+T_t}W_N^{(i)}(y_1^N,\hat{u}_1^{i-1}|g)
\end{align*}
And finally, the message is decoded as $\hat{v}_1^N=\hat{u}_1^N-b_1^N$.\\
The total number of valid input sequences is equal to
\begin{align*}
&2^{NR}=\prod_{t=0}^r |T_t|^{|A_t|}=\prod_{t=0}^r p^{t|A_t|}\approx \prod_{t=0}^r p^{tp_tN}
\end{align*}
Therefore, the rate is equal to $R=\sum_{t=0}^r p_t t\log p$. On the other hand, since $I^n$ is a martingale, we have $\mathds{E}\{I^\infty\}=I^0$. Since $\mathds{E}\{I^\infty\}=\sum_{t=0}^r p_t t\log p$, we observe that the rate $R$ is equal to the symmetric capacity $I^0$. We will see in the next section that this rate is achievable.
\subsection{Error Analysis}
Let $B_i$ be the event that the first error occurs when the decoder decodes the $i$th symbol:
\begin{align}\label{eqn:Bi}
B_i=&\left\{(u_1^N,y_1^N)\in \G^N\times \mathcal{Y}^N|\forall j<i:u_j=f_j(y_1^N, u_1^{j-1}),u_i\ne f_i(y_1^N,u_1^{i-1})\right\}\\
\nonumber &\subseteq \left\{(u_1^N,y_1^N)\in \G^N\times \mathcal{Y}^N|u_i\ne f_i(y_1^N,u_1^{i-1})\right\}
\end{align}%\mathcal{X}_N^\epsilon

For $t=0,\cdots,r$ and $i\in A_t^{\epsilon}$, define
\begin{align}
\nonumber E_i=&\left\{(u_1^N,y_1^N)\in \G^N\times \mathcal{Y}^N|W_N^{(i)}(y_1^N,u_1^{i-1}|u_i)\right.\\
\label{eqn:Ei}&\left.\le W_N^{(i)}(y_1^N,u_1^{i-1}|\tilde{u}_i)\mbox{ for some }\tilde{u}_i\in b_i+T_t\right., \tilde{u}_i\ne u_i\}
\end{align}

\begin{lemma}
For $t=0,\cdots,r$ and $i\in A_t^{\epsilon}$, $P(E_i)\le q^2 Z^t(W_N^{(i)})$.
\end{lemma}

\begin{proof}
For $u_i\in\G$, write $u_i=b_i(u_i)+v_i(u_i)$ where $b_i(u_i)\in H_t$ and $v_i(u_i)\in T_t$. We have
\begin{align*}
P(E_i)&=\sum_{u_1^N,y_1^N}\frac{1}{q^N}W_N(y_1^N|u_1^N)\mathds{1}_{E_i}(u_1^N,y_1^N)\\
&\le \sum_{u_1^N,y_1^N}\frac{1}{q^N}W_N(y_1^N|u_1^N)\sum_{\tilde{u}_i\in b_i(u_i)+T_t,\tilde{u}_i\ne u_i}\sqrt{\frac{W_N^{(i)}(y_1^N,u_1^{i-1}|\tilde{u}_i)}{W_N^{(i)}(y_1^N,u_1^{i-1}|u_i)}}\\
&=\sum_{u_1^i,y_1^N}\frac{1}{q}\left(\sum_{u_{i+1}^N}\frac{1}{q^{N-1}}W_N(y_1^N|u_1^N)\right) \sum_{\tilde{u}_i\in b_i(u_i)+T_t,\tilde{u}_i\ne u_i}\sqrt{\frac{W_N^{(i)}(y_1^N,u_1^{i-1}|\tilde{u}_i)}{W_N^{(i)}(y_1^N,u_1^{i-1}|u_i)}}\\
&=\sum_{u_1^i,y_1^N}\frac{1}{q} W_N^{(i)}(y_1^N,u_1^{i-1}|u_i) \sum_{\tilde{u}_i\in b_i(u_i)+T_t,\tilde{u}_i\ne u_i}\sqrt{\frac{W_N^{(i)}(y_1^N,u_1^{i-1}|\tilde{u}_i)}{W_N^{(i)}(y_1^N,u_1^{i-1}|u_i)}}\\
&= \sum_{u_i\in\G} \sum_{\tilde{u}_i\in b_i(u_i)+T_t,\tilde{u}_i\ne u_i} \frac{1}{q} \sum_{u_1^{i-1},y_1^N}\sqrt{W_N^{(i)}(y_1^N,u_1^{i-1}|\tilde{u}_i)W_N^{(i)}(y_1^N,u_1^{i-1}|u_i)}\\
&= \sum_{u_i\in\G} \sum_{\tilde{u}_i\in b_i(u_i)+T_t,\tilde{u}_i\ne u_i} \frac{1}{q} Z_{\{u_i,\tilde{u}_i\}}(W_N^{(i)})
\end{align*}% Similarly, $u_i-\tilde{u}_i\notin H_s$ for $s=t,t+1,\cdots,r$. K_0\cup K_1\cup\cdots \cup K_{t-1}=

For $u_i\in \G$ and $\tilde{u}_i\in b_i(u_i)+T_t$, if $u_i\ne \tilde{u}_i$, then $u_i,\tilde{u}_i$ are not in the same coset of $H_t$ and hence $u_i-\tilde{u}_i\notin H_t$. Therefore, $u_i-\tilde{u}_i\in \G\backslash H_t$. Note that for $d=u_i-\tilde{u}_i$, $Z_{\{u_i,\tilde{u}_i\}}(W_N^{(i)})\le q Z_d(W_N^{(i)})$. Since $d\in\G\backslash H_t$, we have $Z_d(W_N^{(i)})\le Z^t(W_N^{(i)})$ and hence,
\begin{align*}
Z_{\{u_i,\tilde{u}_i\}}(W_N^{(i)})&\le qZ^t(W_N^{(i)})
\end{align*}
Therefore, $P(E_i)\le q|T_t| Z^t(W_N^{(i)})\le q^2 Z^t(W_N^{(i)})$.
\end{proof}

The probability of block error is given by $P(err) = \sum_{t=0}^r\sum_{i\in A_t^{\epsilon}}P(B_i)$. Since $B_i\subseteq E_i$, we get
\begin{align}\label{Perr_Zpr}
P(err) &\le\sum_{t=0}^r\sum_{i\in A_t^{\epsilon}} q^2 Z^t(W_N^{(i)})\\%\le (r+1) N qZ^t(W_N^{(i)})
&\stackrel{(a)}{\le} \sum_{t=0}^r |A_t^{\epsilon}| q^2 2^{-2^{\beta n}}\\
&\le q^2 N 2^{-2^{\beta n}}
\end{align}
for any $\beta< \frac{1}{2}$ where $(a)$ follows from Remark \ref{remark_rate_of_pol}. Therefore, the probability of error goes to zero as $\epsilon\rightarrow 0$ (and hence $n\rightarrow \infty$).

\section{Polar Codes Over Arbitrary Channels}\label{section:abelian}
For any channel input alphabet there always exist an Abelian group of the same size. In this section, we generalize the result of the previous section to channels of arbitrary input alphabet sizes and arbitrary group operations.
\subsection{Abelian Groups}
Let the Abelian group $\G$ be the input alphabet of the channel. It is a standard fact that any Abelian group can be decomposed into a direct sum of $\mathds{Z}_{p^r}$ rings \cite{algebra_bloch}. Let $\G=\bigoplus_{l=1}^L \R_l$ with $\R_l=\mathds{Z}_{p_l^{r_l}}$ where $p_l$'s are prime numbers and $r_l$'s are positive integers. For $t=(t_1,t_2,\cdots,t_L)$ with $t_l\in\{0,1,\cdots,r_l\}$, there exists a corresponding subgroup $H$ of $\G$ defined by $H=\bigoplus_{l=1}^L p_l^{t_l} \R_l$. For a subgroup $H$ of $\G$ define $T_H$ to be a transversal of $H$ in $\G$.% i.e. $T_H$ is a subset of $G$ containing one and only one element from each coset (shift) of $H$ in $G$.
\subsection{Recursive Channel Transformation}
\subsubsection{The Basic Channel Transforms}
The transformed channels $W^+$ and $W^-$ and the process $I^n(W)$ are defined the same way as the $\mathds{Z}_{p^r}$ case through Equations \eqref{eqn:channel_transform1}, \eqref{eqn:channel_transform2} and \eqref{eqn:Iprocess}.
\subsubsection{Asymptotic Behavior of Synthesized Channels}
For $d\in\G$, define $Z_d^n(W)$ same as (\ref{eqn:Zd}) where $q=|\G|$ and for $H\le \G$, define $I_H^n(W)$ by Equation \eqref{eqn:Iprocess2}. To prove the polarization for arbitrary groups, we need the following lemma:

\begin{lemma}\label{lemma:Zd_sg_gen_by}
For $d_1,d_2\in\G$, if $Z_{d_1}(W)> 1-\epsilon$ and $Z_{d_2}(W)> 1-\epsilon$, then $Z_{\tilde{d}}(W)\approx_{\epsilon} 1$ for any $\tilde{d} \in \langle d_1,d_2 \rangle$ where $\langle d_1,d_2 \rangle$ is the subgroup of $\G$ generated by $d_1$ and $d_2$.
\end{lemma}

\begin{proof}
The condition $Z_{d_1}> 1-\epsilon$ implies $1-Z(W_{\{x,x+d_1\}})\le q\epsilon$ and the condition $Z_{d_2}> 1-\epsilon$ implies $1-Z(W_{\{x,x+d_2\}})\le q\epsilon$. Similar to the proof of Lemma \ref{lemma:k_tk_s}, we have
\begin{align*}
&\sqrt{1-Z(W_{\{x,x+2d_1\}})}\le 2\sqrt{q\epsilon},
&\sqrt{1-Z(W_{\{x,x+2d_2\}})}\le 2\sqrt{q\epsilon}
\end{align*}
It is also straightforward to show that
\begin{align*}
&\sqrt{1-Z(W_{\{x,x+d_1+d_2\}})}\le 2\sqrt{q\epsilon}
\end{align*}
Since $\tilde{d} \in \langle d_1,d_2 \rangle$, it can be written as $\tilde{d}=id_1+jd_2$ for some integers $i,j$. Repeated application of the above inequalities yields the lemma.
\end{proof}

%We have shown that all $Z_d(W)$'s are $0-1$ valued. Let $d_1,d_2,\cdots,d_m$ be the set of all $d$'s with $Z_d=1$. It has been shown in Lemma \ref{lemma:Zd_sg_gen_by} that if $Z_{d_1}(W)\approx 1$ and $Z_{d_2}(W)\approx 1$, then $Z_{d^\prime}(W)\approx 1$ for $d^\prime \in \langle d_1,d_2 \rangle$ where $\langle d_1,d_2 \rangle$ is the subgroup generated by $d_1$ and $d_2$.
\begin{remark}\label{remark:d1_dm}
This lemma is generalizable to the case where for $d_1,\cdots,d_m\in\G$, $Z_{d_1}(W)> 1-\epsilon, Z_{d_2}(W)> 1-\epsilon,\cdots, Z_{d_m}(W)> 1-\epsilon$. In this case, we have $Z_{\tilde{d}}(W)\approx_{\epsilon} 1$ for any $\tilde{d} \in \langle d_1,d_2,\cdots,d_m \rangle$.
\end{remark}

The following lemma is a restatement of Lemma \ref{Zd_0_1_valued}. Here, we prove it for arbitrary groups.

\begin{lemma}\label{Zd_0_1_valued_general}
For all $d\in \G$, $Z^n_d(W)$ converges to a $\{0,1\}$-valued random variable $Z^\infty_d(W)$ as $n$ grows. Moreover, if $\tilde{d}\in\G$ is such that $\langle \tilde{d}\rangle=\langle d\rangle$ then $Z^\infty_{\tilde{d}}(W)=Z^\infty_d(W)$ almost surely; i.e. the random processes $Z^n_{\tilde{d}}(W)$ and $Z^n_d(W)$ converge to the same random variable.
\end{lemma}

\begin{proof}
Similar to the proof of Lemma \ref{Zd_0_1_valued}, let $H=\langle d\rangle$ and let $M$ be any maximal subgroup of $H$. Define
\begin{align}\label{eqn:d_maximal_sg2}
d^\prime=\arg\max_{\substack{a\in H\\a\notin M}} Z_a(W)
\end{align}
It is relatively straightforward to show that in the general case as well, $Z^n_{d'}(W)$ converges to a $\{0,1\}$-valued random variable $Z^\infty_{d'}(W)$. Indeed this part of the proof of Lemma \ref{Zd_0_1_valued} is general enough for arbitrary Abelian groups. Here we show that this implies $Z^n_d(W)$ also converges to a Bernoulli random variable.\\ \\
Let $|H|=\prod_{i=1}^k q_i^{a_i}$ where $q_i$'s are distinct primes and $a_i$'s are positive integers. Note that $H$ is isomorphic to the cyclic group $\mathds{Z}_{|H|}$. For $i=1,\cdots,k$, define the subgroup $M_i=\langle q_i\rangle$ of $\mathds{Z}_{|H|}$ (and isomorphically of $H$) and let $d'_i=\arg\max_{\substack{a\in H\\a\notin M_i}} Z_a(W)$. Note that for $i=1,\cdots,k$, $M_i$ is a maximal subgroup of $\mathds{Z}_{|H|}$ (and isomorphically of $H$). Therefore, for $i=1,\cdots,k$, $Z^n_{d'_i}(W)$ converges to a $\{0,1\}$-valued random variable. If for some $i=1,\cdots,k$, $Z_{d'_i}(W)<\epsilon$ it follows that $Z_d(W)<\epsilon$ (since $d\in H\backslash M_i$) and if for all $i=1,\cdots,k$, $Z_{d'_i}(W)>1-\epsilon$, it follows from Remark \ref{remark:d1_dm} that $Z_{\tilde{d}}(W)>1-O(\epsilon)$ for any $\tilde{d}\in\langle d'_1,d'_2,\cdots,d'_k\rangle$. Next, we show that $\langle d'_1,d'_2,\cdots,d'_k\rangle=H$ and this will prove that if for all $i=1,\cdots,k$, $Z_{d'_i}(W)>1-\epsilon$ then $Z_{d}(W)>1-O(\epsilon)$. For $i=1,\cdots,k$, since $d'_i\notin M_i$ it follows that $d'_i \not \equiv 0 \pmod{q_i}$. Define
\begin{align*}
\delta=\sum_{i=1}^k \left(\prod_{\substack{j=1\\j\ne i}}^k q_j\right) d'_i
\end{align*}
Then we have $\delta\not\equiv 0 \pmod{q_i}$ for all $i=1,\cdots,k$. This implies $\langle \delta\rangle=H$ and hence $\langle d'_1,d'_2,\cdots,d'_k\rangle=H$. Therefore, if in the limit $Z_{d'_i}(W)=0$ for some $i=1,\cdots,k$ then $Z_d(W)=0$ and if $Z_{d'_i}(W)=0$ for all $i=1,\cdots,k$ then $Z_d(W)=1$. This proves that $Z_d^n(W)$ converges to a Bernoulli random variable.\\ \\
If $\tilde{d}\in\G$ is such that $\langle \tilde{d}\rangle=\langle d\rangle$ then it follows that $\tilde{d}\in H$ and $\tilde{d}\notin M_i$ for $i=1,\cdots,k$. Therefore if in the limit $Z_{d'_i}(W)=0$ for some $i=1,\cdots,k$ then $Z_{\tilde{d}}(W)=0$ and if $Z_{d'_i}(W)=0$ for all $i=1,\cdots,k$ then $Z_{\tilde{d}}(W)=1$. This proves that the random processes $Z^n_{\tilde{d}}(W)$ and $Z^n_d(W)$ converge to the same random variable.
%If $\tilde{d}\in\G$ is such that $\langle \tilde{d}\rangle=\langle d\rangle$, then
\end{proof}

%For $d\in\G$ define $H=\langle d\rangle$ and let $M_1\subseteq M_2\subseteq\cdots\subseteq M_t=H$ be a chain of subgroups of $H$ such that $M_s$ is maximal in $M_{s+1}$. For $s=2,\cdots,t$ define $d_s'=\arg\max_{a\in M_s\backslash M_{s-1}} Z_a(W)$.  We have shown in Lemma \ref{Zd_0_1_valued} that $Z_{d_s'}(W)$ converges to a Bernoulli random variable.\\
%Note that $M_1$ is necessarily a $\mathds{Z}_p$ field. Hence Lemma \ref{Zd_0_1_valued} implies that for all $d\in M_1$, $d\ne 0$, $Z_d(W)$'s converge to the same random variable. Assume $Z_{d_2'}(W)$ is close to $0$ then for all $d\in M_2\backslash M_1$, $Z_d(W)$ is close to zero. If $Z_{d_2'}(W)$ is close to one, then there exists a maximal subgroup of $M_2$ whose

%\begin{lemma}\label{lemma:Zd_sg_gen_by}
%If $Z_{d_1}(W)<\epsilon$ and $Z_{d_2}(W)<\epsilon$, then $Z_{d^\prime}(W)\approx 0$ for $d^\prime \in \langle d_1,d_2 \rangle$ where $\langle d_1,d_2 \rangle$ is the subgroup of $G$ generated by $d_1$ and $d_2$.
%\end{lemma}
%
%\begin{proof}
%Immediate from Lemma \ref{lemma:Zd1} and Lemma \ref{lemma:Zd1_2}.
%\end{proof}

In the asymptotic regime, let $d_1,d_2,\cdots,d_m$ be all elements of $\G$ such that $Z_{d_i}(W)=1$ and assume that for all other elements $d\in \G$, $Z_d(W)= 0$ (we can make this assumption since in the limit $Z_d$'s are $\{0,1\}$-valued). We have seen that if $Z_{d_i}(W)=1$ for $i=1,\cdots,m$ then for any $\tilde{d}\in\langle d_1,d_2,\cdots,d_m\rangle$, $Z_{\tilde{d}}(W)=1$. Therefore, $\langle d_1,d_2,\cdots,d_m\rangle\subseteq \{d_1,d_2,\cdots,d_m\}$ and hence $\{d_1,d_2,\cdots,d_m\}=\langle d_1,d_2,\cdots,d_m\rangle=H$ for some subgroup $H$ of $\G$. This means all possible asymptotic cases can be indexed by subgroups of $\G$. i.e. for any $H\le \G$, one possible asymptotic case is
\begin{itemize}
\item \textbf{Case $H$: }$Z_d(W)=\left\{ \begin{array}{ll}
         1 & \mbox{if $d\in H$};\\
        0 & \mbox{Otherwise}.\end{array} \right.$
\end{itemize}
where for $H\le \G$, case $H$ happens with some probability $p_H$.\\
Next, We study the behavior of $I^n$ in each of these cases.

\begin{lemma}
For a channel $(\G,\mathcal{Y},W)$ and for a subgroup $S$ of $\G$, if $Z_d>1-\epsilon$ for $d\in S$ and $Z_d<\epsilon$ for $d\notin S$, then $I^0(W)\approx_{\epsilon}\log\frac{|\G|}{|S|}$.
\end{lemma}

\begin{proof}
Let $0=M_0\subseteq M_1\subseteq \cdots \subseteq M_{t-1} \subseteq S=M_{t}\subseteq M_{t+1}\subseteq \cdots \G=M_{k}$ for some positive integer $k$ be any chain of subgroups such that $M_{s-1}$ is maximal in $M_{s}$ for $s=1,\cdots,k$.\\ \\
For $s=1,\cdots,t$ let $H=M_s$ and $M=M_{s-1}$ and let $T_{H}$ be a transversal of $H$ in $\G$ and let $T_{M}$ be a transversal of $M$ in $H$. For $d\in H$, we have $Z_d(W)>1-\epsilon$. For $t_{H}\in T_{H}$ define the channel $\bar{W}(y|t_H+t_M+M_{s-1})$ similar to \eqref{eq:W_bar_def}. We have shown in Appendix \ref{section:Lower_Bound_ZWbar1} that if for some $d\in H\backslash M$, $Z_d(W)>1-\epsilon$ then $Z_{d+t_H+M}(\bar{W})>1-O(\epsilon)$. Since the input alphabet of the channel $\bar{W}$ has a prime size, we can use \cite[Lemma 4]{sasoglu_polar_q} to conclude that $Z(\bar{W})>1-O(\epsilon)$. Now we use \cite[Prop. 3]{sasoglu_polar_q} to conclude $I(\bar{W})<O(\epsilon)$. This result is valid for all $t_H\in T_H$. Since $I(\bar{W})=I(\hat{\tilde{X}};Y|\hat{X}=t_H)$, we conclude that
\begin{align*}
I_H(W)-I_M(W)&=\sum_{t_H\in T_H} P(\hat{X}=t_H) I(\hat{\tilde{X}};Y|\hat{X}=t_H)\\
&<O(\epsilon)
\end{align*}
Therefore, for $s=1,\cdots,t$, $I_{M_s}(W)-I_{M_{s-1}}(W)\approx_{\epsilon}0$ and hence, $I_{M_t}(W)=I_S(W)\approx_\epsilon I_{M_0}(W)=0$.\\ \\
For $s=t+1,\cdots,k$ let $H=M_s$ and $M=M_{s-1}$ and let $T_{H}$ be a transversal of $H$ in $\G$ and let $T_{M}$ be a transversal of $M$ in $H$. For $d\in H\backslash M$, we have $Z_d(W)<\epsilon$. For the channel $\bar{W}$ defined as above, we have shown in Appendix \ref{section:Upper_Bound_ZWbar} that if for all $d\in H\backslash M$, $Z_d(W)<\epsilon$ then $Z(\bar{W})<O(\epsilon)$. Therefore, \cite[Prop. 3]{sasoglu_polar_q} implies $I(\bar{W})=\log \frac{|H|}{|M|}-O(\epsilon)$. Similar as above, we conclude that
\begin{align*}
I_H(W)-I_M(W)&=\log \frac{|H|}{|M|}-O(\epsilon)
\end{align*}
Therefore, for $s=t+1,\cdots,k$, $I_{M_s}(W)-I_{M_{s-1}}(W)\approx_{\epsilon}\log\frac{|M_s|}{|M_{s-1}|}$ and hence
\begin{align*}
I_{\G}(W)-I_{S}(W)&\approx_\epsilon \sum_{s=t+1}^k\log \frac{|M_s|}{|M_{s-1}|}\\
&=\log\frac{|\G|}{|S|}
\end{align*}
Since $I_S(W)\approx_\epsilon 0$, We conclude that $I^0(W)=I_{\G}(W)\approx_\epsilon \log\frac{|\G|}{|S|}$.
%Considering the continuity of the mutual information, it suffices to show the lemma in the limiting case where $Z_d=1$ for $d\in H$ and $Z_d=0$ for $d\notin H$. From Lemma \ref{lemma:Zd1_22}, we have $W(y|x)=W(y|\tilde{x})$ if $x-\tilde{x}\in H$ and from Lemma \ref{lemma:Zd0} we have $W(y|x)W(y|\tilde{x})=0$ if $x-\tilde{x}\in \G\backslash H_t$. Therefore for $y\in\mathcal{Y}$ with positive probability $p_Y(y)$, $W(y|x)$ is uniform over a coset of $H$ and zero over all other cosets. i.e. for all $y\in\mathcal{Y}$, there exists a coset $C_H^y$ of $H$ such that $\frac{1}{q}W(y|x)=\frac{p_Y(y)}{|C_H^y|}=\frac{p_Y(y)}{\prod_{l=1}^L p^{r-t_l}}$ for $y\in C_H^y$ and $W(y|x)=0$ otherwise. The mutual information is equal to
%\begin{align*}
%I^0(W)&=\sum_{x\in\G}\sum_{y\in\mathcal{Y}}\frac{1}{q}W(y|x)\log\frac{W(y|x)}{\sum_{\tilde{x} \in\mathcal{X}}\frac{1}{q}W(y|\tilde{x})}\\
%&=\sum_{y\in\mathcal{Y}} \sum_{x\in C_H^y} \frac{1}{q}W(y|x)\log\frac{W(y|x)}{\sum_{\tilde{x} \in C_H^y}\frac{1}{q}W(y|\tilde{x})}\\
%&=\sum_{y\in\mathcal{Y}} \sum_{x\in C_H^y} \frac{p_Y(y)}{\prod_{l=1}^L p_l^{r-t_l}}\log\frac{\frac{\prod_{l=1}^L p_l^{r_l} p_Y(y)}{\prod_{l=1}^L p_l^{r-t_l}}}{\sum_{\tilde{x} \in C_H^y}\frac{p_Y(y)}{p_l^{r-t_l}}}\\
%&=\sum_{y\in\mathcal{Y}} p_Y(y)\log \prod_{l=1}^L p_l^{t_l}=\sum_{l=1}^L t_l\log p_l=\log \frac{|\G|}{|H|}
%\end{align*}
\end{proof}

We have shown that the process $I^n$ converges to the following discrete random variable: $I^{\infty}=\log\frac{|\G|}{|H|}$ with probability $p_H$ for $H\le \G$.\\
For $H\le \G$, define the random variable $Z^H(W_N^{(i)})=\sum_{d\notin H} Z_d(W_N^{(i)})$ and the random process $(Z^H)^{(n)}(W)=Z^H(W_N^{(J_n)})$ where $J_n$ is a uniform random variable over $\{1,2,\cdots,N=2^n\}$. Note that $(Z^H)^{(n)}(W)$ converges almost surely to a random variable $(Z^H)^{(\infty)}(W)$ and $P\left((Z^H)^{(\infty)}=0\right)=\sum_{S\le H}p_S$.\\
\subsubsection{Summary of Channel Transformation}
For the channel $(\G,\mathcal{Y},W)$, the convergence of the processes $I^n$ and $(Z^H)^n$ for $H\le \G$ implies that for all $\epsilon>0$, there exists a number $N=N(\epsilon)$ and a partition $\{A_H^\epsilon|H\le \G\}$ of $\{1,\cdots,N\}$ such that for $H\le \G$ and $i\in A_H^\epsilon$, $I(W_N^{(i)})=\log\frac{|\G|}{|H|}+O(\epsilon)$ and $Z^H(W_N^{(i)})=O(\epsilon)$. Moreover, as $\epsilon\rightarrow 0$, $\frac{|A_H^\epsilon|}{N}\rightarrow p_H$ for some probabilities $p_H, H\le \G$.\\
In Appendix \ref{section:Rate_of_Polarization}, we show that for any $\beta<\frac{1}{2}$ and for $H\le \G$,
\begin{align}\label{rate_of_pol}
\lim_{n\rightarrow \infty}P\left((Z^H)^{(n)}<2^{-2^{\beta n}}\right)&\ge P\left((Z^H)^{(\infty)}=0\right)\\
\nonumber &=\sum_{S\le H}^r p_S
\end{align}
This implies that for all $\epsilon>0$, there exists a number $N=N(\epsilon)=2^{n(\epsilon)}$ and a partition 
$\{A_H^\epsilon|H\le \G\}$ of $\{1,\cdots,N\}$ such that for $H\le \G$ and $i\in A_H^\epsilon$, $I(W_N^{(i)})=\log\frac{|\G|}{|H|}+O(\epsilon)$ and $Z^H(W_N^{(i)})<2^{-2^{\beta n(\epsilon)}}$. Moreover, as $\epsilon\rightarrow 0$, $\frac{|A_H^\epsilon|}{N}\rightarrow p_H$ for some probabilities $p_H, H\le \G$.
\subsection{Encoding and Decoding}
At the encoder, if $i\in A_H^\epsilon$ for some $H\le \G$, the information symbol is chosen from the transversal $T_H$ arbitrarily. Let $\mathcal{X}_N^\epsilon=\bigoplus_{H\le \G} T_H^{A_H^\epsilon}$ be the set of all valid input sequences. As in the $\mathds{Z}_{p^r}$ case, the message $u_1^N$ is dithered with a uniformly distributed random vector $b_1^N\in\bigoplus_{H\le \G} H^{A_H^\epsilon}$ revealed to both the encoder and the decoder. A message $v_1^N\in\mathcal{X}_N^\epsilon$ is encoded to the vector $x_1^N=(v_1^N+b_1^N) G_N$. Note that $u_1^N=v_1^N+b_1^N$ is uniformly distributed over $\G^N$.\\
At the decoder, after observing the output vector $y_1^N$, for $H\le \G$ and $i\in A_H^\epsilon$ , use the following decoding rule:
\begin{align*}
\hat{u}_i =f_i(y_1^N,\hat{u}_1^{i-1})=\argmax_{g\in b_i+T_H}W_N^{(i)}(y_1^N,\hat{u}_1^{i-1}|g)
\end{align*}

And finally, the message is recovered as $\hat{v}_1^N=\hat{u}_1^N-b_1^N$.\\
The total number of valid input sequences is equal to
\begin{align*}
&2^{NR}=\prod_{H\le \G} |T_H|^{|A_H|}=\prod_{H\le \G} \left(\frac{|\G|}{|H|}\right)^{|A_H|}
\end{align*}

Therefore the rate is equal to $R=\sum_{H\le \G} \frac{|A_H|}{N}\log\frac{|\G|}{|H|}$. On the other hand, since $I^n$ is a martingale, we have $\mathds{E}\{I^\infty\}=I^0$. Since $\mathds{E}\{I^\infty\}=\sum_{H\le \G} p_H \log\frac{|\G|}{|H|}$, we observe that the rate $R$ converges to the symmetric capacity $I^0$ as $\epsilon\rightarrow 0$. We will see in the next section that this rate is achievable.

\subsection{Error Analysis}
For $H\le G$ and $i\in A_H^\epsilon$, define the events $B_i$ and $E_i$ according to Equations \eqref{eqn:Bi} and \eqref{eqn:Ei}. 
%Let $B_i$ be the event that the first error occurs when the decoder decodes the $i$th symbol:
%\begin{align*}
%B_i=&\left\{(u_1^N,y_1^N)\in \G^N\times \mathcal{Y}^N|\forall j<i:u_j=f_j(y_1^N,u_1^{j-1}),u_i\ne f_i(y_1^N,u_1^{i-1})\right\}\\
%&\subseteq \left\{(u_1^N,y_1^N)\in \G^N\times \mathcal{Y}^N|u_i\ne f_i(y_1^N,u_1^{i-1})\right\}
%\end{align*}
%
%For $H\le G$ and $i\in A_H^\epsilon$, define
%\begin{align*}
%E_i=&\left\{(u_1^N,y_1^N)\in \G^N\times \mathcal{Y}^N|W_N^{(i)}(y_1^N,u_1^{i-1}|u_i)\right.\\
%&\left.\le W_N^{(i)}(y_1^N,u_1^{i-1}|\tilde{u}_i)\mbox{ for some }\tilde{u}_i\in T_H,\tilde{u}_i\ne u_i\right\}
%\end{align*}
Similar to the $\mathds{Z}_{p^r}$ case, it is straightforward to show that for $H\le G$ and $i\in A_H^\epsilon$, $P(E_i)\le q^2 Z^H(W_N^{(i)})$ where $q=|\G|$. The probability of block error is given by $P(err) = \sum_{H\le \G} \sum_{i\in A_H^\epsilon}P(B_i)$. Since $B_i\subseteq E_i$, we get
\begin{align*}
P(err) &\le\sum_{H\le \G} \sum_{i\in A_H^\epsilon} q^2 Z^H(W_N^{(i)})\\
&\le \sum_{H\le \G} |A_H^\epsilon| q^2 2^{-2^{\beta n}}\\
&\le q^2 N 2^{-2^{\beta n}}
\end{align*}
for any $\beta< \frac{1}{2}$. Therefore, the probability of block error goes to zero as $\epsilon\rightarrow 0$ ($n\rightarrow \infty$).

\section{Relation to Group Codes}\label{section:examples}
Recall that for an arbitrary group $\G$, the polar encoder of length $N$ introduced in this paper maps the set $\bigoplus_{H\le \G}T_H^{A_H}$ to $\G^N$ where for a subgroup $H$ of $\G$, $T_H$ is a transversal of $H$ and $\{A_H|H\le \G\}$ is some partition of $\{1,\cdots,N\}$. Note that the set of messages $\bigoplus_{H\le \G}T_H^{A_H}$ is not necessarily closed under addition and hence in general, the set of encoder outputs is not a subgroup of $\G^N$; i.e. polar codes constructed and analyzed in Sections \ref{rings} and \ref{section:abelian} are not group encoders. On the contrary, the standard polar codes (i.e. polar codes in which only perfect channels are used) are indeed group codes since their set of messages is of the form $\G^A\oplus \{0\}^{\{1,\cdots,N\}\backslash A}$ for some $A\subseteq \{1,\cdots,N\}$ which is closed under addition.\\

%It is known that group codes generally do not achieve the symmetric capacity of discrete memoryless channels. Hence, one could predict that standard polar codes cannot achieve the symmetric capacity of arbitrary channels and a modification of the encoding rule is indeed necessary to achieve that goal.\\

It is worth mentioning that polar encoders constructed in this paper fall into a larger class of structured codes called \emph{nested group codes}. Nested group codes consist of two group codes: the inner code $\mathds{C}_i$ and the outer code $\mathds{C}_o$ such that the inner code is a subgroup of the outer code ($\mathds{C}_i\le \mathds{C}_o$). The set of messages consists of cosets of $\mathds{C}_i$ in $\mathds{C}_o$. For the case of polar codes, the inner code is given by
\begin{align*}
\mathds{C}_i&=\left[\bigoplus_{H\le \G}H^{A_H}\right]G\\
&=\left\{mG\left|m\in \bigoplus_{H\le \G}H^{A_H}\right.\right\}
\end{align*}
and the outer code is the whole group space: $\mathds{C}_o=\G^N$. To verify that this is indeed the case, it suffices to show that the set of codewords of polar codes $\left[\bigoplus_{H\le \G}T_H^{A_H}\right]G$ has only one common element with each coset of $\mathds{C}_i$. Equivalently, it suffices to show that for $m_1,m_2\in\G^N$, if $m_1G-m_2G\in \mathds{C}_i$, then either $m_1\notin \bigoplus_{H\le \G}T_H^{A_H}$ or $m_2\notin \bigoplus_{H\le \G}T_H^{A_H}$.
\begin{lemma}
For $N=2^n$ where $n$ is a positive integer, the generator matrix corresponding to polar codes $G_N=B_NF^{\otimes n}$ is full rank.
\end{lemma}
\begin{proof}
Since $G_N=B_NF^{\otimes n}$ where $B_N$ is a permutation of rows, it suffices to show that $F^{\otimes n}$ is full rank. Note that the rank of the Kronecker product of two matrices is equal to the product of the ranks of matrices and the rank of $F$ is equal to $2$. Hence we have $\mbox{rank}(G)=\mbox{rank}(F^{\otimes n})=2^n=N$.
\end{proof}

This lemma implies that if $m_1G-m_2G\in \mathds{C}_i$ then $m_1-m_2\in \bigoplus_{H\le \G}H^{A_H}$. This means either $m_1\notin \bigoplus_{H\le \G}T_H^{A_H}$ or $m_2\notin \bigoplus_{H\le \G}T_H^{A_H}$. This proves that polar codes are indeed nested group codes.\\

In this section, we consider two examples of channels over $\mathds{Z}_4$. The first example is Channel 1 introduced in Section \ref{section:example}. Based on the symmetry of this channel, we show that polar codes achieve the group capacity of this specific channel. The intent of the second example is to show that in general, polar codes do not achieve the group capacity of channels. In order to find the capacity of polar codes as group codes, we use the standard construction of polar codes, i.e. we only use perfect channels and fix partially perfect and useless channels.
\subsection{Example 1}
Consider Channel 1 of Figure \ref{fig:channel}. Define $H_0=\{0,1,2,3\}$, $H_1=\{0,2\}$ and $H_2=\{0\}$ and define $K_0=\{1,3\}$, $K_1=\{2\}$ and $K_2=\{0\}$. For this channel we have:
\begin{align*}
&I^0\triangleq I(X;Y)=2-\epsilon-2\lambda\\
&I_2^0\triangleq I(X_1;Y)=1-(\epsilon+\lambda)\\
&(I_2^\prime)^0\triangleq I(X_1';Y)=1-(\epsilon+\lambda)=I_2^0
\end{align*}
where $X$ is uniform over $\mathds{Z}_4$, $X_1$ is uniform over $H_1$ and $X_1'$ is uniform over $1+H_1$. The capacity of group codes over this symmetric channel is equal to \cite{abelianp2p_ieee2}:
\begin{align*}
C&=\min(I_4^0,I_2^0+(I_2^\prime)^0)=\min(2-\epsilon-2\lambda,2-2\epsilon-2\lambda)\\
&=2-2\epsilon-2\lambda
\end{align*}

All possible cases for this channel are
\begin{itemize}
\item \textbf{Case 0:} $Z_1^\infty=Z_3^\infty=1,Z_2^\infty=1$
\item \textbf{Case 1:} $Z_1^\infty=Z_3^\infty=0,Z_2^\infty=1$
\item \textbf{Case 2:} $Z_1^\infty=Z_3^\infty=0,Z_2^\infty=0$
\end{itemize}

As we saw in Figures \ref{fig:pol2} and \ref{fig:pol}, this result agrees with the asymptotic behavior of $I^n$ predicted by the recursion formulas (\ref{eqn:recursion1}) and (\ref{eqn:recursion2}). \\

Define $I(W^{b_1b_2\cdots b_n})=I(X;Y)$ where $X$, $Y$ are the input and output of $W^{b_1b_2\cdots b_n}$ and $X$ is uniform over $\mathds{Z}_4$. Similarly, define $I_2(W^{b_1b_2\cdots b_n})=I(X_1;Y)$ where $X_1$, $Y$ are the input and output of $W^{b_1b_2\cdots b_n}$ and $X_1$ is uniform over $H_1$ and define $I_2^\prime(W^{b_1b_2\cdots b_n})=I(X_1';Y)$ where $X_1'$, $Y$ are the input and output of $W^{b_1b_2\cdots b_n}$ and $X_1'$ is uniform over $1+H_1$. Define the mutual information processes $I_4^n$, $I_2^n$ and $(I_2')^n$ to be equal to $I(W^{b_1b_2\cdots b_n})$, $I_2(W^{b_1b_2\cdots b_n})$ and $I_2^\prime(W^{b_1b_2\cdots b_n})$ where for $i=1,\cdots,n$, $b_i$'s are iid Bernoulli$(0.5)$ random variables. For this channel, we can show that $I_2(W^{b_1b_2\cdots b_n})=I_2^\prime(W^{b_1b_2\cdots b_n})=1-(\epsilon_n+\lambda_n)$ and conclude that $(I_2+I_2^\prime)^n\triangleq I_2^n+(I_2^\prime)^n$ is a martingale. Therefore $I_4^n$ and $(I_2+I_2^\prime)^n$ converge almost surely to random variables $I_4^{\infty}$ and $(I_2+I_2^\prime)^{\infty}$ respectively. This observation provides us with an ad-hoc way to find the probabilities $p_t$, $t=0,1,2$ of the limit random variable $I_4^\infty$ for this simple channel. We can show the following for the final states:
\begin{itemize}
\item \textbf{case 0} $\Rightarrow I_4^\infty=0, (I_2+I_2^\prime)^\infty=0$
\item \textbf{case 1} $\Rightarrow I_4^\infty=1, (I_2+I_2^\prime)^\infty=0$
\item \textbf{case 2} $\Rightarrow I_4^\infty=2, (I_2+I_2^\prime)^\infty=2$
\end{itemize}

Therefore we obtain the following three equations:
\begin{align*}
&\mathds{E}\{I_4^\infty\}=p_0\cdot 0+p_1\cdot 1+p_2\cdot 2=I_4^0=2-\epsilon-2\lambda\\
&\mathds{E}\{(I_2+I_2^\prime)^\infty\}=p_0\cdot 0+p_1\cdot 0+p_2\cdot 2=(I_2+I_2^\prime)^0=2-2\epsilon-2\lambda\\
&p_0+p_1+p_2=1
\end{align*}

Solving this system of equations, we obtain:
\begin{align*}
&p_2=1-\epsilon-\lambda=C/2\\
&p_1=I_4^0-(I_2+I_2^\prime)^0\\
&p_0=1-\left(I_4^0-(I_2+I_2^\prime)^0/2\right)
\end{align*}

We see that the fraction of perfect channels is equal to the capacity of the channel achievable using group codes and therefore, polar codes achieve the capacity of group codes for this channel.
\subsection{Example 2}
The channel is depicted in Figure \ref{fig:channel2}. We call This Channel 3.
\begin{figure}[h]
\centering
\includegraphics[scale=.55]{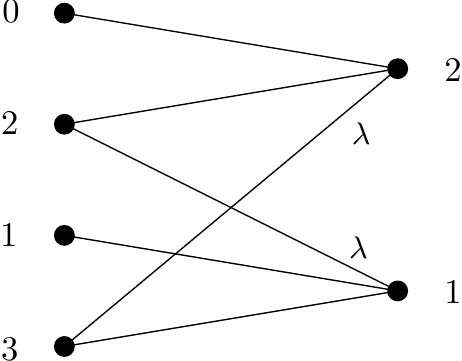}
\caption{\small Channel 3}
\label{fig:channel2}
\end{figure}
For this channel, when $\lambda=0.2$ we have:
\begin{align*}
&I^0=I(X;Y)=0.6390\\
&(I_2^0+I_2^\prime)^0=0.2161
\end{align*}
The rate $C=\min(I_4^0,(I_2+I_2^\prime)^0)=(I_2+I_2^\prime)^0=0.2161$ is achievable using group codes over this channel \cite{abelianp2p_ieee2}.\\
For this channel we have three possible asymptotic case:
\begin{itemize}
\item \textbf{Case 0:} $Z_1^\infty=1,Z_2^\infty=1\Rightarrow I_4^\infty=0, (I_2+I_2^\prime)^\infty=0$
\item \textbf{Case 1:} $Z_1^\infty=0,Z_2^\infty=1\Rightarrow I_4^\infty=1, (I_2+I_2^\prime)^\infty=0$
\item \textbf{Case 2:} $Z_1^\infty=0,Z_2^\infty=0\Rightarrow I_4^\infty=2, (I_2+I_2^\prime)^\infty=2$
\end{itemize}
Therefore we obtain the following three equations:
\begin{align*}
&\mathds{E}\{I_4^\infty\}=p_0\cdot 0+p_1\cdot 1+p_2\cdot 2\\
&\mathds{E}\{(I_2+I_2^\prime)^\infty\}=p_0\cdot 0+p_1\cdot 0+p_2\cdot 2\\
&p_0+p_1+p_2=1
\end{align*}
Therefore, the achievable rate using polar codes over this channel is equal to $R=2p_2=\mathds{E}\{(I_2+I_2^\prime)^\infty\}$. We have $\mathds{E}\{(I_2+I_2^\prime)^1\}=0.2063$ which is strictly less than $(I_2+I_2^\prime)^0$. The following lemma implies $R=\mathds{E}\{(I_2+I_2^\prime)^\infty\}\le \mathds{E}\{(I_2+I_2^\prime)^1\} <C=(I_2+I_2^\prime)^0$ and completes the proof.

\begin{lemma}%\label{lemma:supermartingale}
For a channel $(\mathds{Z}_4,\mathcal{Y},W)$, the process $(I_2+I^\prime_2)^n, n=0,1,2,\cdots$ is a super-martingale.
\end{lemma}
\begin{proof}
Follow from Lemma \ref{lemma:supermartingale} with $H=\{0,2\}$.
\end{proof}

\section{Conclusion}
It has been shown that the original construction of polar codes suffices to achieve the symmetric capacity of discrete memoryless channels with arbitrary input alphabet sizes. It is shown that in general, channel polarization happens in several levels so that some synthesized channels are partially perfect and there needs to be a modification of the coding scheme to exploit these channels. It has also been shown that polar codes do not generally achieve the capacity of arbitrary channels achievable using group codes.
%\bibliographystyle{plain}
%\bibliography{ariabib}

%\pagebreak
\appendix
%\section{Appendices}%\href{}{}
\subsection{Polar Codes Over Abelian Groups}\label{section:polar_abelian}
Given a $k\times n$ matrix $G_n$ of $0$'s and $1$'s, one can construct a group code as follows: Given any message tuple $u^k\in G^k$, encode it to $u^k\cdot G_n$. Where the elements of $G_n$ determine whether an element of $u^k$ appears as a summand in the encoded word or not. For example consider the generator matrix
\begin{align*}
G_4=
\left( \begin{array}{cccc}
1 & 0 & 0 & 0\\
1 & 0 & 1 & 0\\
1 & 1 & 0 & 0\\
1 & 1 & 1 & 1\end{array} \right)
\end{align*}
Then $u^4\cdot G_4$ is defined as
\begin{align*}
[u_1 u_2 u_3 u_4]\cdot
\left( \begin{array}{cccc}
1 & 0 & 0 & 0\\
1 & 0 & 1 & 0\\
1 & 1 & 0 & 0\\
1 & 1 & 1 & 1\end{array} \right)
=\left( \begin{array}{c}
u_1+u_2+u_3+u_4\\
u_3+u_4\\
u_2+u_4\\
u_4\end{array} \right)
\end{align*}
Using this convention, we can define a group code based on a given binary matrix without actually defining a multiplication operation for the group.

\subsection{Recursion Formula for Channel 1}\label{appendix:recursion}
\subsubsection{Recursion for $W^+$}
We show that $W^+$ (corresponding to $b_1=1$) is equivalent to a channel of the same type as $W$ but with different parameters $\epsilon_1$ and $\lambda_1$ corresponding to $\epsilon$ and $\lambda$ respectively; where,
\begin{align*}
&\epsilon_1=\epsilon^2+2\epsilon\lambda\\
&\lambda_1=\lambda_1^2
\end{align*}
We say an output tuple $(y_1,y_2,u_1)$ is connected to an input $u_2\in \mathds{Z}_4$ if $W^+(y_1,y_2,u_1|u_2)=\frac{1}{4}W(y_1|u_1+u_2)W(y_2|u_2)$ is strictly positive.\\

First, let us assume the output tuple $(y_1,y_2,u_1)$ is connected to all $u_2\in \mathds{Z}_4$. Then $W(y_2|u_2)$ must be nonzero for all $u_2$ and hence $y_2=E_3$. Similarly since $W(y_1|u_1+u_2)$ is nonzero for all $u_2$ (and hence all $u_1+u_2$) it follows that $y_1=E_3$. Therefore $W^+(E_3,E_3,u_1|u_2)=\frac{1}{4}\lambda^2$ for all $u_1,u_2\in \mathds{Z}_4$ and these are all output tuples connected to all inputs (with positive probability). Since all of these output tuples are equivalent we can combine them to get a single output symbol connected to all four inputs with probability $\lambda^2$.\\

Next we show that if an output tuple is connected to an input from $\{0,2\}$ and an input from $\{1,3\}$, then it is connected to all inputs. Consider the case where the output tuple $(y_1,y_2,u_1)$ is connected to both $0$ and $1$ i.e. $W^+(y_1,y_2,u_1|0)$ and $W^+(y_1,y_2,u_1|1)$ are both nonzero. Then since $W(y_2|0)\ne 0$ and $W(y_2|1)\ne 0$, it follows that $y_2=E_3$. Similarly since $W(y_1|u_1)\ne 0$ and $W(y_1|u_1+1)\ne 0$, it follows that $y_1=E_3$. We have already seen that for all $u_1\in \mathds{Z}_4$, the output tuple $(E_3,E_3,u_1)$ is connected to all input symbols. The proof is similar for other three cases i.e. when $(y_1,y_2,u_1)$ is connected to $0$ and $3$, when $(y_1,y_2,u_1)$ is connected to $2$ and $1$, and when $(y_1,y_2,u_1)$ is connected to $2$ and $3$.\\

Next we find all output tuples which are connected to both $0$ and $2$ but are not connected to $1$ or $3$. Let $(y_1,y_2,u_1)$ be an output tuple such that $W^+(y_1,y_2,u_1|0)\ne 0$, $W^+(y_1,y_2,u_1|2)\ne 0$, $W^+(y_1,y_2,u_1|1)= 0$ and $W^+(y_1,y_2,u_1|3)= 0$.\\
First assume $u_1\in\{0,2\}$. Since $W(y_2|0)\ne 0$ and $W(y_2|2)\ne 0$, it follows that $y_2\in \{E_1,E_3\}$ and since $W(y_1|u_1)\ne 0$ and $W(y_1|u_1+2)\ne 0$, it follows that $y_1\in \{E_1,E_3\}$. Note that for $y_1=E_3$ and $y_3=E_3$, the output tuple is connected to all inputs and therefore all possible cases are $y_1=E_1, y_2=E_1$, $y_1=E_1, y_2=E_3$ and $y_1=E_3, y_2=E_1$. In all cases it can be shown that $W^+(y_1,y_2,u_1|1)= 0$ and $W^+(y_1,y_2,u_1|3)= 0$. Hence for $u_1\in\{0,2\}$, $(E_1,E_1,u_1)$ is connected to $0$ and $2$ with probabilities $\frac{1}{4}\epsilon^2$ and is not connected to $1$ or $3$. $(E_1,E_3,u_1)$ is connected to $0$ and $2$ with probabilities $\frac{1}{4}\epsilon\lambda$ and is not connected to $1$ or $3$. $(E_3,E_1,u_1)$ is connected to $0$ and $2$ with probabilities $\frac{1}{4}\epsilon\lambda$ and is not connected to $1$ or $3$.\\
Now assume $u_1\in\{1,3\}$. Same as above we have $y_2\in \{E_1,E_3\}$ and since $W(y_1|u_1)\ne 0$ and $W(y_1|u_1+2)\ne 0$, it follows that $y_1\in \{E_2,E_3\}$. In this case, all possible cases are $y_1=E_2, y_2=E_1$, $y_1=E_2, y_2=E_3$ and $y_1=E_3, y_2=E_1$. In all cases it can be shown that $W^+(y_1,y_2,u_1|1)= 0$ and $W^+(y_1,y_2,u_1|3)= 0$. Hence for $u_1\in\{1,3\}$, $(E_2,E_1,u_1)$ is connected to $0$ and $2$ with probabilities $\frac{1}{4}\epsilon^2$ and is not connected to $1$ or $3$. $(E_2,E_3,u_1)$ is connected to $0$ and $2$ with probabilities $\frac{1}{4}\epsilon\lambda$ and is not connected to $1$ or $3$. $(E_3,E_1,u_1)$ is connected to $0$ and $2$ with probabilities $\frac{1}{4}\epsilon\lambda$ and is not connected to $1$ or $3$.\\
Therefore, there are four equivalent outputs connected to $0$ and $2$ with probabilities $\frac{1}{4}\epsilon^2$ and not connected to $1$ or $3$ and there are eight equivalent outputs connected to $0$ and $2$ with probabilities $\frac{1}{4}\epsilon\lambda$ and not connected to $1$ or $3$. Since all of these outputs are equivalent, we can combine them into one output connected to $0$ and $2$ with probabilities
\begin{align*}
4\left(\frac{1}{4}\epsilon^2\right)+8\left(\frac{1}{4}\epsilon\lambda\right)=\epsilon^2+2\epsilon\lambda
\end{align*}

Now we find all output tuples which are connected to both $1$ and $3$ but are not connected to $0$ or $2$. Let $(y_1,y_2,u_1)$ be an output tuple such that $W^+(y_1,y_2,u_1|1)\ne 0$, $W^+(y_1,y_2,u_1|3)\ne 0$, $W^+(y_1,y_2,u_1|0)= 0$ and $W^+(y_1,y_2,u_1|2)= 0$.\\
First assume $u_1\in\{0,2\}$. Since $W(y_2|1)\ne 0$ and $W(y_2|3)\ne 0$, it follows that $y_2\in \{E_2,E_3\}$ and since $W(y_1|u_1+1)\ne 0$ and $W(y_1|u_1+3)\ne 0$, it follows that $y_1\in \{E_2,E_3\}$. Note that for $y_1=E_3$ and $y_3=E_3$, the output tuple is connected to all inputs and therefore all possible cases are $y_1=E_2, y_2=E_2$, $y_1=E_2, y_2=E_3$ and $y_1=E_3, y_2=E_2$. In all cases it can be shown that $W^+(y_1,y_2,u_1|0)= 0$ and $W^+(y_1,y_2,u_1|2)= 0$. Hence for $u_1\in\{0,2\}$, $(E_2,E_2,u_1)$ is connected to $1$ and $3$ with probabilities $\frac{1}{4}\epsilon^2$ and is not connected to $0$ or $2$. $(E_2,E_3,u_1)$ is connected to $1$ and $3$ with probabilities $\frac{1}{4}\epsilon\lambda$ and is not connected to $0$ or $2$. $(E_3,E_2,u_1)$ is connected to $1$ and $3$ with probabilities $\frac{1}{4}\epsilon\lambda$ and is not connected to $0$ or $2$.\\
Now assume $u_1\in\{1,3\}$. Same as above we have $y_2\in \{E_2,E_3\}$ and since $W(y_1|u_1+1)\ne 0$ and $W(y_1|u_1+3)\ne 0$, it follows that $y_1\in \{E_1,E_3\}$. In this case, all possible cases are $y_1=E_1, y_2=E_2$, $y_1=E_1, y_2=E_3$ and $y_1=E_3, y_2=E_2$. In all cases it can be shown that $W^+(y_1,y_2,u_1|0)= 0$ and $W^+(y_1,y_2,u_1|2)= 0$. Hence for $u_1\in\{1,3\}$, $(E_1,E_2,u_1)$ is connected to $1$ and $3$ with probabilities $\frac{1}{4}\epsilon^2$ and is not connected to $0$ or $2$. $(E_1,E_3,u_1)$ is connected to $1$ and $3$ with probabilities $\frac{1}{4}\epsilon\lambda$ and is not connected to $0$ or $2$. $(E_3,E_2,u_1)$ is connected to $1$ and $3$ with probabilities $\frac{1}{4}\epsilon\lambda$ and is not connected to $0$ or $2$.\\
Therefore, there are four equivalent outputs connected to $1$ and $3$ with probabilities $\frac{1}{4}\epsilon^2$ and not connected to $0$ or $2$ and there are eight equivalent outputs connected to $1$ and $3$ with probabilities $\frac{1}{4}\epsilon\lambda$ and not connected to $0$ or $2$. Same as above, since all of these outputs are equivalent, we can combine them into one output connected to $1$ and $3$ with probabilities $\epsilon^2+2\epsilon\lambda$.\\

We have shown that there is (equivalently) one channel output (call it $E_3^+$) connected to all inputs $u_2\in\mathds{Z}_4$ with conditional probability $\lambda_1=\lambda^2$ and we have shown that if a channel output is connected to more that one input but is not connected to all inputs, it is either connected to $\{0,2\}$ and is not connected to $\{1,3\}$ (call it $E_1^+$) or it is connected to $\{0,2\}$ and is not connected to $\{1,3\}$ (call it $E_2^+$). $0$ and $2$ are connected to $E_1^+$ with probabilities $\epsilon_1=\epsilon^2+2\epsilon\lambda$ and $1$ and $3$ are connected to $E_2^+$ with probabilities $\epsilon_1=\epsilon^2+2\epsilon\lambda$. Then for each input $u_2\in \mathds{Z}_4$ these exist several outputs which are only connected to $u_2$ and not other inputs and whose sum of probabilities add up to $1-\epsilon_1-\lambda_1$. This completes the proof for $W^+$.\\

\subsubsection{Recursion for $W^-$}
We show that $W^-$ (corresponding to $b_1=0$) is equivalent to a channel of the same type as $W$ but with different parameters $\epsilon_1$ and $\lambda_1$ corresponding to $\epsilon$ and $\lambda$ respectively; where,
\begin{align*}
&\epsilon_1=2\epsilon-\left(\epsilon^2+2\epsilon\lambda\right)\\
&\lambda_1=2\lambda-\lambda_1^2
\end{align*}
Note that each channel output is a pair $(y_1,y_2)\in\{0,1,2,3,E_1,E_2,E_3\}^2$. The channel $W^-$ can be shown to be as following:\\
Output pairs $(0,0)$, $(1,1)$, $(2,2)$, $(3,3)$ are only connected to input $0$ each with conditional probability $\frac{1}{4}(1-\epsilon-\lambda)^2$. This is equivalent to one channel output only connected to $0$ with probability $(1-\epsilon-\lambda)^2$.\\
Output pairs $(0,2)$, $(1,3)$, $(2,0)$, $(3,1)$ are only connected to input $2$ each with conditional probability $\frac{1}{4}(1-\epsilon-\lambda)^2$. This is equivalent to one channel output only connected to $2$ with probability $(1-\epsilon-\lambda)^2$.\\
Output pairs $(0,3)$, $(1,0)$, $(2,1)$, $(3,2)$ are only connected to input $1$ each with conditional probability $\frac{1}{4}(1-\epsilon-\lambda)^2$. This is equivalent to one channel output only connected to $1$ with probability $(1-\epsilon-\lambda)^2$.\\
Output pairs $(0,1)$, $(1,2)$, $(2,3)$, $(3,0)$ are only connected to input $3$ each with conditional probability $\frac{1}{4}(1-\epsilon-\lambda)^2$. This is equivalent to one channel output only connected to $3$ with probability $(1-\epsilon-\lambda)^2$.\\
Output pairs $(0,E_1)$, $(1,E_2)$, $(2,E_1)$, $(3,E_2)$, $(E_1,0)$, $(E_1,2)$, $(E_2,1)$, $(E_2,3)$ are only connected to inputs $0$ and $2$ each with conditional probability $\frac{1}{4}\epsilon(1-\epsilon-\lambda)$. Output pairs $(E_1,E_1)$, $(E_2,E_2)$ are only connected to inputs $0$ and $2$ each with conditional probability $\frac{1}{2}\epsilon^2$. This is equivalent to one channel output only connected to $0$ and $2$ with probability
\begin{align*}
\epsilon_1&=8\times\frac{1}{4}\epsilon(1-\epsilon-\lambda)+2\times \frac{1}{2}\epsilon^2\\
&=2\epsilon-\left(\epsilon^2+2\epsilon\lambda\right)
\end{align*}
Output pairs $(0,E_2)$, $(1,E_1)$, $(2,E_2)$, $(3,E_1)$, $(E_1,1)$, $(E_1,3)$, $(E_2,0)$, $(E_2,2)$ are only connected to inputs $1$ and $3$ each with conditional probability $\frac{1}{4}\epsilon(1-\epsilon-\lambda)$. Output pairs $(E_1,E_2)$, $(E_2,E_1)$ are only connected to inputs $1$ and $3$ each with conditional probability $\frac{1}{2}\epsilon^2$. This is equivalent to one channel output only connected to $1$ and $3$ with probability $2\epsilon-\left(\epsilon^2+2\epsilon\lambda\right)$.\\
Output pairs $(0,E_3)$, $(1,E_3)$, $(2,E_3)$, $(3,E_3)$, $(E_3,0)$, $(E_3,1)$, $(E_3,2)$, $(E_3,3)$ are  connected to all inputs each with conditional probability $\frac{1}{4}\lambda(1-\epsilon-\lambda)$. Output pairs $(E_1,E_3)$, $(E_2,E_3)$, $(E_3,E_1)$, $(E_3,E_2)$ are connected to all inputs each with conditional probability $\frac{1}{2}\epsilon\lambda$. Output pair $(E_3,E_3)$ is connected to all inputs with conditional probability $\lambda^2$. This is equivalent to one channel output only connected to all inputs with probability
\begin{align*}
\epsilon_1&=8\times\frac{1}{4}\lambda(1-\epsilon-\lambda)+4\times \frac{1}{2}\epsilon\lambda+\lambda^2\\
&=2\lambda-\lambda^2
\end{align*}
We have listed all $49$ channel outputs and the corresponding probabilities. This completes the proof for $W^-$.\\
\subsection{Upper Bound on $Z(\bar{W})$}\label{section:Upper_Bound_ZWbar}
Assume $Z_{d^\prime}(W)<\epsilon$. This implies
\begin{align*}
\frac{1}{q}\sum_{x\in\G}\sum_{y\in\mathcal{Y}}\sqrt{W(y|x)W(y|x+\tilde{d})}< \epsilon
\end{align*}
for all $\tilde{d}\in H\backslash M$. Therefore for each $x\in\G$,
\begin{align}\label{eqn:upper_bound_on_B}
\sum_{y\in\mathcal{Y}}\sqrt{W(y|x)W(y|x+\tilde{d})}< q\epsilon
\end{align}
The Bhattacharyya parameter of the channel $\bar{W}$ is given by:
\begin{align*}
Z(\bar{W})&=\frac{1}{\bar{q}(\bar{q}-1)} \sum_{\substack{t_M,t_M'\in T_M\\t_M\ne t_M'}} \sum_{y\in\mathcal{Y}} \sqrt{\bar{W}(y|t_H+t_M+M)\bar{W}(y|t_H+t_M'+M)}\\
&=\frac{1}{\bar{q}(\bar{q}-1)}\frac{1}{|M|} \sum_{\substack{t_M,t_M'\in T_M\\t_M\ne t_M'}} \sum_{y\in\mathcal{Y}} \sqrt{\left(\sum_{m\in M}W(y|t_H+t_M+m)\right)\left(\sum_{m'\in M}W(y|t_H+t_M'+m')\right)}\\
&=\frac{1}{\bar{q}(\bar{q}-1)}\frac{1}{|M|} \sum_{\substack{t_M,t_M'\in T_M\\t_M\ne t_M'}} \sum_{y\in\mathcal{Y}} \sqrt{\sum_{m,m'\in M}W(y|t_H+t_M+m)W(y|t_H+t_M'+m')}\\
&\le \frac{1}{\bar{q}(\bar{q}-1)}\frac{1}{|M|} \sum_{\substack{t_M,t_M'\in T_M\\t_M\ne t_M'}} \sum_{y\in\mathcal{Y}} \sum_{m,m'\in M} \sqrt{W(y|t_H+t_M+m)W(y|t_H+t_M'+m')}\\
\end{align*}
Let $x=t_H+t_M+m$ and $x'=t_H+t_M'+m'$. Note that $x-x'=t_M-t_M'+m-m'\in H$ since $t_M,t_M',m,m'\in H$. Also note that since $t_M\ne t_M'$ and $m-m'\in M$, it follows that $x-x'\notin M$. Now we use \eqref{eqn:upper_bound_on_B} to conclude:
\begin{align*}
Z(\bar{W})&\le \frac{1}{\bar{q}(\bar{q}-1)}\frac{1}{|M|} \sum_{\substack{t_M,t_M'\in T_M\\t_M\ne t_M'}} \sum_{m,m'\in M} q\epsilon\\
&\le \frac{1}{\bar{q}(\bar{q}-1)}\frac{1}{|M|} (\frac{|H|}{|M|})^2 |M|^2 q\epsilon
=\frac{|M|\cdot |H|\cdot |G|}{|H|-|M|}\epsilon
\end{align*}

\begin{remark}
For an arbitrary Abelian group $\G$, let $H\le \G$ be an arbitrary subgroup and let $M$ be any maximal subgroup of $H$. If for all $\tilde{d}\in H\backslash M$, $Z_{\tilde{d}}(W)<\epsilon$ then with a similar argument as above we can show that $Z(\bar{W})<O(\epsilon)$ where $\bar{W}$ is defined by \eqref{eq:W_bar_def}.
\end{remark}
\subsection{Lower Bound on $Z_{d'+t_H+M}(\bar{W})$}\label{section:Lower_Bound_ZWbar1}
Assume $Z_{d'}(W)>1-\epsilon$. Define
\begin{align*}
D_{d'}(W)=\frac{1}{2q}\sum_{x\in\G}\sum_{y\in\mathcal{Y}}\left|W(y|x)-W(y|x+d')\right|
\end{align*}
First we show that $Z_{d'}(W)>1-\epsilon$ implies $D_{d'}(W)<O(\epsilon)$. Define the following quantities:
\begin{align*}
q_{x,y}=\frac{W(y|x)+W(y|x+d')}{2}\\
\delta_{x,y}=\frac{1}{2}\left|W(y|x)-W(y|x+d')\right|
\end{align*}
Then we have
\begin{align*}
Z_{d'}(W)&=\frac{1}{q} \sum_{x\in\G} \sum_{y\in\mathcal{Y}} \sqrt{(q_{x,y}-\delta_{x,y}) (q_{x,y}+\delta_{x,y})}\\
&=\frac{1}{q} \sum_{x\in\G} \sum_{y\in\mathcal{Y}} \sqrt{q^2_{x,y}-\delta^2_{x,y}}
\end{align*}
Also we have
\begin{align*}
&D\triangleq\frac{1}{q}\sum_{x\in\G} \sum_{y\in\mathcal{Y}}\delta_{x,y}=D_{d'}(W),
\end{align*}
and
\begin{align*}
&0\le \delta_{x,y}\le q_{x,y}
\end{align*}
Note that
\begin{align*}
Z_{d'}(W)\le \max_{\substack{d_{x,y}:\frac{1}{q}\sum_{x\in\G} \sum_{y\in\mathcal{Y}} d_{x,y}=D}} \frac{1}{q} \sum_{x\in\G} \sum_{y\in\mathcal{Y}} \sqrt{q^2_{x,y}-d^2_{x,y}}
\end{align*}
The Lagrangian for this optimization problem is given by
\begin{align*}
\mathcal{L}=\frac{1}{q} \sum_{x\in\G} \sum_{y\in\mathcal{Y}} \sqrt{q^2_{x,y}-d^2_{x,y}} -\lambda \left(\frac{1}{q} \sum_{x\in\G} \sum_{y\in\mathcal{Y}} d_{x,y}-D\right)
\end{align*}
we have
\begin{align*}
\frac{\partial}{\partial d_{x,y}}\mathcal{L}=-\frac{d_{x,y}}{\sqrt{q^2_{x,y}-d^2_{x,y}}}- \frac{\lambda}{q}
\end{align*}
and
\begin{align*}
\frac{\partial^2}{\partial d_{x,y}^2}\mathcal{L}=-\frac{q^2_{x,y}}{(q^2_{x,y}-d^2_{x,y})^{\frac{3}{2}}}\le 0
\end{align*}
Define $\gamma=- \frac{\lambda}{q}$ to get $d_{x,y}=\sqrt{\frac{\gamma^2}{1+\gamma^2}}q_{x,y}$. We have $\sum_{y\in\mathcal{Y}}q_{x,y}=1$, therefore,
\begin{align*}
\frac{1}{q} \sum_{x\in\G} \sum_{y\in\mathcal{Y}} d_{x,y}&=\frac{1}{q} \sum_{x\in\G} \sum_{y\in\mathcal{Y}} \sqrt{\frac{\gamma^2}{1+\gamma^2}}q_{x,y}\\
&=\sqrt{\frac{\gamma^2}{1+\gamma^2}} \frac{1}{q} \sum_{x\in\G} \sum_{y\in\mathcal{Y}} q_{x,y}\\
&=\sqrt{\frac{\gamma^2}{1+\gamma^2}}
\end{align*}
Therefore we have $D=\sqrt{\frac{\gamma^2}{1+\gamma^2}}$ and hence $d_{x,y}=D q_{x,y}$. For this choice of $d_{x,y}$ we have
\begin{align*}
\frac{1}{q} \sum_{x\in\G} \sum_{y\in\mathcal{Y}} \sqrt{q^2_{x,y}-d^2_{x,y}}&=\frac{\sqrt{1-D^2}}{q} \sum_{x\in\G} \sum_{y\in\mathcal{Y}} q_{x,y}\\
&=\sqrt{1-D^2}
\end{align*}
Therefore, we have shown that $Z_{d'}(W)\le\sqrt{1-D_{d'}(W)^2}$. This implies that $D_{d'}(W)<2\epsilon-\epsilon^2=O(\epsilon)$.\\

Next, we show that $D_{d'}(W)<\epsilon$ implies $Z_{d'}(W)>1-O(\epsilon)$. We need the following lemma:
\begin{lemma}\label{lemma:ab2}
For constants $0\le a\le b\le 1$, with $b-a\le \delta$,
\begin{align*}
\sqrt{ab}\ge \frac{a+b}{2}-\frac{\delta}{2}
\end{align*}
\end{lemma}
\begin{proof}
%The proof is similar to the proof of Lemma \ref{ab}.
Note that
\begin{align*}
\frac{a+b}{2}-\sqrt{ab}\le \max_{0\le x-a\le\delta}\frac{a+x}{2}-\sqrt{ax}
\end{align*}
We have
\begin{align*}
\frac{\partial}{\partial x}\left[\frac{a+x}{2}-\sqrt{ax}\right]=\frac{1}{2}-\frac{a}{2\sqrt{ax}}\ge 0
\end{align*}
for all $x\ge a$. Therefore the maximum is attained at $x=a+\delta$. Therefore,
\begin{align*}
\frac{a+b}{2}-\sqrt{ab}\le \frac{a+(a+\delta)}{2}-\sqrt{a(a+\delta)}
\end{align*}
The maximum of the right hand side is attained at $a=0$, hence,
\begin{align*}
\frac{a+b}{2}-\sqrt{ab}\le \frac{\delta}{2}
\end{align*}
\end{proof}

Assume $D_{d'}(W)<\epsilon$. We have
\begin{align*}
1-Z_{d'}(W)&=1-\frac{1}{q} \sum_{x\in\G} \sum_{y\in\mathcal{Y}} \sqrt{W(y|x)W(y|x+d')}\\
&=\frac{1}{q} \sum_{x\in\G} \sum_{y\in\mathcal{Y}} \left(\frac{W(y|x)+W(y|x+d')}{2}-\sqrt{W(y|x)W(y|x+d')}\right)\\
&\stackrel{(a)}{\le} \frac{1}{q} \sum_{x\in\G} \sum_{y\in\mathcal{Y}} \frac{1}{2}\left|W(y|x)-W(y|x+d')\right|\\
&=D_{d'}(W)
\end{align*}
where $(a)$ follows from Lemma \ref{lemma:ab2} with $a=W(y|x)$, $b=W(y|x+d')$ and $\delta=\left|W(y|x)-W(y|x+d')\right|$. This shows that $D_{d'}(W)<\epsilon$ implies $Z_{d'}(W)>1-\epsilon$.\\ \\

Next, we show that $D_{d'}(W)<\epsilon$ implies $D_{d'+t_H+M}(\bar{W})<O(\epsilon)$. We have
\begin{align*}
D_{d'+t_H+M}(\bar{W})&=\frac{1}{\bar{2q}} \sum_{t_M\in T_M} \sum_{y\in\mathcal{Y}} \left|\bar{W}(y|t_H+t_M+M)-\bar{W}(y|t_H+t_M+d'+M)\right|\\
&=\frac{1}{\bar{q}} \frac{1}{|M|} \sum_{t_M\in T_M} \sum_{y\in\mathcal{Y}} \left|\sum_{m\in M}W(y|t_H+t_M+m)-\sum_{m\in M}W(y|t_H+t_M+d'+m)\right|\\
&\le \frac{1}{\bar{q}} \frac{1}{|M|} \sum_{t_M\in T_M} \sum_{y\in\mathcal{Y}} \sum_{m\in M} \left|W(y|t_H+t_M+m)-W(y|t_H+t_M+d'+m)\right|\\
&\le \frac{1}{\bar{q}} \frac{1}{|M|} 2q D_{d'}(W)
\end{align*}
This shows that $D_{d'}(W)<\epsilon$ implies $D_{d'+t_H+M}(\bar{W})<\frac{2q\epsilon}{\bar{q}|M|}=O(\epsilon)$.\\

We have shown that $Z_{d'}(W)>1-\epsilon$ implies $D_{d'}(W)<2\epsilon-\epsilon^2=O(\epsilon)$. This implies $D_{d'+t_H+M}(\bar{W})<\frac{2q(2\epsilon-\epsilon^2))}{\bar{q}|M|}=O(\epsilon)$ and this in turn implies $Z_{d'+t_H+M}(\bar{W})>1-\frac{2q(2\epsilon-\epsilon^2)}{\bar{q}|M|}=1-O(\epsilon)$.

\begin{remark}
For an arbitrary Abelian group $\G$, let $H\le \G$ be an arbitrary subgroup and let $M$ be any maximal subgroup of $H$. If for some $\tilde{d}\in H\backslash M$, $Z_{\tilde{d}}(W)>1-\epsilon$ then with a similar argument as above, we can show that $Z_{\tilde{d}+t_H+M}(\bar{W})>1-O(\epsilon)$ where $\bar{W}$ is defined by \eqref{eq:W_bar_def}.
\end{remark}
\subsection{Alternate Proof for a Lower Bound on $Z_{d'+t_H+M}(\bar{W})$} \label{section:Alt_Lower_Bound_ZWbar}
In Appendix \ref{section:Lower_Bound_ZWbar1}, we proved that $Z_{d'}(W)>1-\epsilon$ implies $Z_{d'+t_H+M}(\bar{W})>1-O(\epsilon)$. In this part, we give an alternate proof of this statement for the $\mathds{Z}_{p^r}$ case.\\
Assume $Z_{d'}(W)>1-\epsilon$. It follows that
\begin{align*}
\sum_{x\in\G}\left[1-\sum_{y\in\mathcal{Y}}\sqrt{W(y|x)W(y|x+d')}\right]<q\epsilon
\end{align*}
Similar to the previous case, we have for all $x\in\G$,
\begin{align*}
\sqrt{1-Z(W_{\{x,x+2d'\}})}\le 2\sqrt{q\epsilon}
\end{align*}
Repeated application of the above lemma yields $\forall x,x'\in\G:x-x'\in \langle d'\rangle$,
\begin{align}\label{eqn:Wxxprime}
\sqrt{1-Z(W_{\{x,x'\}})}\le q\sqrt{q\epsilon}
\end{align}
We have
\begin{align*}
Z_{d'+t_H+M}(\bar{W})&=\frac{1}{\bar{q}}\sum_{t_M\in T_M}\sum_{y\in\mathcal{Y}} \sqrt{\bar{W}(y|t_H+t_M+M)\bar{W}(y|t_H+t_M+d'+M)}\\
&=\frac{1}{\bar{q}} \sum_{t_M\in T_M} \sum_{y\in\mathcal{Y}} \sqrt{\sum_{m,m'\in M}\frac{1}{|M|^2}W(y|t_H+t_M+m)W(y|t_H+t_M+d'+m')}\\
&\stackrel{(a)}{\ge} \frac{1}{\bar{q}} \sum_{t_M\in T_M} \sum_{y\in\mathcal{Y}} \sum_{m,m'\in M} \frac{1}{|M|^2} \sqrt{W(y|t_H+t_M+m)W(y|t_H+t_M+d'+m')}\\
&\ge \frac{1}{\bar{q}} \sum_{t_M\in T_M} \min_{m,m'\in M} \sum_{y\in\mathcal{Y}} \sqrt{W(y|t_H+t_M+m)W(y|t_H+t_M+d'+m')}
\end{align*}
where $(a)$ follows since $\sqrt{\cdot}$ is a concave function. Let $x=t_H+t_M+m$ and $x'=t_H+t_M+d'+m'$. It follows that $x'-x=d'+(m'-m)$. Since $d',m',m\in H$ we have $x'-x\in H$. Since $\G$ and hence $H$ are $\mathds{Z}_{p^r}$ rings it follows that $d'\in H\backslash M$ generates $H$; hence $x'-x\in \langle d'\rangle$. We can use \eqref{eqn:Wxxprime} to get
\begin{align*}
Z_{d'+t_H+M}(\bar{W})&\ge \frac{1}{\bar{q}} \sum_{t_M\in T_M} \min_{m,m'\in M} (1-q^3\epsilon)=1-\frac{q^3 \epsilon}{\bar{q}}
\end{align*}
It follows that $Z_{d'+t_H+M}(\bar{W})>1-O(\epsilon)$.
\subsection{The Rate of Polarization}\label{section:Rate_of_Polarization}
Recall that for $t=0,\cdots,r$, $(Z^t)^{(n)}=\sum_{d\notin H_t}Z_d(W_N^{(J_n)})$ where $J_n$ is uniform over $\{1,2,\cdots,2^n\}$. For $t=0,\cdots,r$, define $(Z^t_{\max})^{(n)}=\max_{d\notin H_t}Z_d(W_N^{(J_n)})$ where $J_n$ is same as above. Since for all $d\in \G$, $Z_d(W^+)=Z_d(W)^2$ it follows that $Z^t_{\max}(W^+)\le Z^t_{\max}(W)^2$. It has been shown in \cite[p. 6]{sasoglu_polar_q} that
\begin{align*}
Z_d(W^-)\le 2Z_d(W)+\sum_{\substack{\Delta\ne 0\\\Delta\ne -d}}Z_{\Delta}(W)Z_{d+\Delta}(W)
\end{align*}
Note that for any $\Delta\in G$, $d\notin H_t$ implies that either $\Delta\notin H_t$ or $d+\Delta\notin H_t$. Therefore, $d\notin H_t$ implies either $Z_{\Delta}(W)\le Z^t_{\max}(W)$ or $Z_{d+\Delta}(W)\le Z^t_{\max}(W)$ (or both). Since $Z_{\Delta}(W)$ and $Z_{d+\Delta}(W)$ both take values from $[0,1]$, it follows that
\begin{align*}
Z_{\Delta}(W)Z_{d+\Delta}(W)\le Z^t_{\max}(W)
\end{align*}
Therefore, for any $d\notin H_t$, $Z_d(W^-)\le 2Z_d(W)+q Z^t_{\max}(W)$. Hence
\begin{align*}
Z^t_{\max}(W^-)&=\max_{d\notin H_t} Z_d(W^-)\\
&\le \max_{d\notin H_t} \left(2Z_d(W)+ q Z^t_{\max}(W)\right)\\
&\le (q+2)Z^t_{\max}(W)
\end{align*}
Since for all $d$ $Z_d^n$ converges to a Bernoulli random variable it follows that $(Z^t_{\max})^{(n)}$ also converges to a $\{0,1\}$-valued random variable $(Z^t_{\max})^{(\infty)}$. Note that $P\left((Z^t_{\max})^{(\infty)}=0\right)=P\left((Z^t)^{\infty}=0\right)=\sum_{s=t}^r p_s$. Therefore, $(Z^t_{\max})^{(n)}$ satisfies the conditions of \cite[Theorem 1]{Arikan_rate_of_pol_ieee} and hence
\begin{align*}
\lim_{n\rightarrow \infty} P\left((Z^t_{\max})^{(n)}<2^{-2^{\beta n}}\right)=P\left((Z^t_{\max})^{(\infty)}=0\right)
\end{align*}
for any $\beta<\frac{1}{2}$. It clearly follows that $\lim_{n\rightarrow \infty} P\left(q(Z^t_{\max})^{(n)}<2^{-2^{\beta n}}\right)=P\left((Z^t_{\max})^{(\infty)}=0\right)$. Note that the event $\{(Z^t)^{(n)}<2^{-2^{\beta n}}\}$ includes the event $\{q(Z^t_{\max})^{(n)}<2^{-2^{\beta n}}\}$. Therefore,
\begin{align*}
\lim_{n\rightarrow \infty} P\left((Z^t)^{(n)}<2^{-2^{\beta n}}\right)\ge P\left((Z^t)^{\infty}=0\right)
\end{align*}

Similarly, for an arbitrary Abelian group $\G$ and a subgroup $H$ of $\G$, define $(Z^H_{\max})^{(n)}=\max_{d\notin H}Z_d(W_N^{(J_n)})$ where $J_n$ is defined as above. It is straightforward to show that $(Z^H_{\max})^{(n)}$ satisfies the conditions of \cite[Theorem 1]{Arikan_rate_of_pol_ieee}. Therefore, with an argument similar to above, we can show that,
\begin{align*}
\lim_{n\rightarrow \infty} P\left((Z^H)^{(n)}<2^{-2^{\beta n}}\right)\ge P\left((Z^H)^{\infty}=0\right)
\end{align*}
for any $\beta<\frac{1}{2}$.
%\bibliographystyle{plain}
%\bibliography{ariabib}
\bibliographystyle{IEEEtran}
\bibliography{IEEEabrv,ariabib}
\end{document}